\documentclass{IEEEtran}
\usepackage[utf8]{inputenc}
\usepackage{cite}
\usepackage{bm}
\usepackage{mathtools,amsmath,amssymb,amsthm}
\usepackage{tikz}
\usepackage{pgfplots}
\usepackage{xcolor}
\usepackage{myTikzStyle}
\usepgfplotslibrary{ternary}
\usepackage[printonlyused,nolist]{acronym}
\usepackage{myAbreviations}
\usepackage{algorithm}
\usepackage{algpseudocode}

\usepackage{graphicx}
\usepackage{subcaption}
\usepackage{comment}

\newcommand{\mset}[1]{\ensuremath{\mathcal{#1}}}
\newcommand{\rv}[1]{\ensuremath{\mathsf{#1}}}

\newcommand{\done}[2]{\ensuremath{d_1\left( #1 , #2 \right)}}
\newcommand{\E}[2]{\ensuremath{\mathbb{E}_{#1}\left[#2\right]}}
\newcommand{\Var}[2]{\ensuremath{\mathbb {V}_{#1}\left[#2\right]}}

\newcommand{\selfinfo}[2]{\iota_{#1}\left( #2 \right)}
\newcommand{\entrp}[1]{\ensuremath{\mathbb{H}\left(#1\right)}}
\newcommand{\crossentrp}[2]{\ensuremath{\mathbb{X}\left( #1 \Vert #2 \right)}}
\newcommand{\diverg}[2]{\ensuremath{\mathbb{D}\left( #1 \Vert #2 \right)}}

\newcommand{\Rinfo}[0]{\ensuremath{R_{\text{info}}}}
\newcommand{\Rrng}[0]{\ensuremath{R_{\text{rng}}}}
\newcommand{\Brng}[0]{\ensuremath{B_{\text{rng}}}}
\newcommand{\Hrng}[0]{\ensuremath{H_{\text{rng}}}}

\newcommand{\bl}[0]{\ensuremath{n}}
\newcommand{\ptarget}[0]{\ensuremath{Q_{\rv{A}}}}
\newcommand{\pblock}[0]{\ensuremath{P_{\rv{A}^n}}}
\newcommand{\binset}[1]{\ensuremath{\mset{S}_{#1}}}

\newcommand{\numW}[0]{\ensuremath{K}}

\newcommand{\pmf}[1]{\ensuremath{P_{\rv{#1}}}}
\newcommand{\qmf}[1]{\ensuremath{Q_{\rv{#1}}}}

\newcommand{\nup}[0]{N_{\uparrow}}
\newcommand{\ndwn}[0]{N_{\downarrow}}
\newcommand{\qup}[0]{q_{\uparrow}}
\newcommand{\qdwn}[0]{q_{\downarrow}}
\DeclareMathOperator{\supp}{supp}

\newcommand{\psplotheight}[0]{2.5 in}

\newcommand{\unif}[1]{U_{#1}}

\newtheorem{theorem}{Theorem}
\newtheorem{lemma}{Lemma}
\newtheorem{proposition}{Proposition}

\newtheorem{definition}{Definition}

\theoremstyle{definition}

\newtheorem{example}{Example}

\title{Invertible Low-Divergence Coding}
\date{}
\author{Patrick Schulte,~\IEEEmembership{Member,~IEEE}, Rana Ali Amjad,~\IEEEmembership{Member,~IEEE}, Thomas Wiegart,~\IEEEmembership{Student Member,~IEEE}, and Gerhard Kramer,~\IEEEmembership{Fellow,~IEEE}

\thanks{
Date of current version \today.
This work was supported by the German Research Foundation (DFG) through project KR 3517/9-1.
}

\thanks{
Patrick Schulte was with the Chair of Communications Engineering, Technical University of Munich (TUM), 80333 Munich, Germany. He is now with the Huawei Munich Research Center, 80992 Munich, Germany (e-mail: patrick.schulte1@huawei.com).

Gerhard Kramer and Thomas Wiegart are with the Chair of Communications Engineering, Technical University of Munich (TUM), 80333 Munich, Germany (e-mail: gerhard.kramer@tum.de; thomas.wiegart@tum.de).

Rana Ali Amjad was with the Chair of Communications Engineering, Technical University of Munich (TUM), 80290 Munich, Germany (e-mail: ranaali.amjad@tum.de).
}
}

\begin{document}

\maketitle

\begin{abstract}
Several applications in communication, control, and learning require approximating target distributions to within small informational divergence (I-divergence).
The additional requirement of invertibility usually leads to using encoders that are one-to-one mappings, also known as distribution matchers. However, even the best one-to-one encoders have I-divergences that grow logarithmically with the block length in general.
To improve performance, an encoder is proposed that has an invertible one-to-many mapping and a low-rate resolution code. Two algorithms are developed to design the mapping by assigning strings in either a most-likely first or least-likely first order. Both algorithms give information rates approaching the entropy of the target distribution with exponentially decreasing I-divergence and with vanishing resolution rate in the block length. 
\end{abstract}

%-------------------
\section{Introduction}
\label{sec:introduction}
Approximating target distributions has applications such as energy-efficient communication, random number generators, distributed control, coordination, learning, stealth, and others. We are motivated by applications that require both good distribution matching and invertibility. For example, \ac{I-divergence}-minimization and invertibility are useful for variational inference~\cite{Jordan-1999,hafner2020action} and image processing~\cite{liu2020deep}. Another example is the stealth communication problem~\cite{hou2014effective,hou2017effective,lentner2020} where two parties try to hide communication from a ``warden''. The model has two possible states: one party sends either ``noise'' with per-letter statistics $\qmf{A}$ or it sends a string $a^n=a_1,\ldots,a_n$ of symbols that carries a message but resembles the noise. The warden observes $a^n$ and makes a hypothesis test. One finds that if the \ac{I-divergence} of the noise and message statistics is zero, then the best that the warden can do is to guess. 

We are interested in block codes and encoders that:
\begin{enumerate}
    \item map uniformly-distributed messages to strings $a^n$;
    \item transmit messages at rate near the entropy $\entrp{\qmf{A}}$;
    \item exhibit vanishing \ac{I-divergence} in the block length $n$;
    \item permit recovering the transmitted message from $a^n$.
\end{enumerate}
The last requirement suggests that the encoder should be a one-to-one mapping. However, invertibility makes the problem trickier than usual. For example, we find that:
\begin{itemize}
    \item \ac{DM} encoders such as \ac{CCDM} \cite{bocherer2015bandwidth,schulte2016constant} or shell mapping \cite{laroia1994optimal,khandani1993shaping}
    have rates that approach $\entrp{\qmf{A}}$ from \emph{below} and are one-to-one mappings, see Fig.~\ref{fig:o2o}; decoding is therefore invertible but the \ac{I-divergence} grows with $n$~\cite{schulte2017divergence};
    %%%%%%%%%%
    \item \acp{RNG}~\cite{Vonneumann-NBS51,han1997interval} or \acp{RC}~\cite{wyner1975common,han1993approximation} have rates that approach $\entrp{\qmf{A}}$ from \emph{above} and \acp{I-divergence} that vanish with $n$; however, the encoders are many-to-one mappings, see Fig.~\ref{fig:m2o}, and decoding is not invertible.
\end{itemize}
We refer to~\cite{han1997interval,boecherer-geiger-IT16} for more discussion and references on the relations between \ac{DM} and \ac{RC}/\ac{RNG} and their applications to, e.g., shaping for communication. Results for learning, stealth, control, and coordination are developed and reviewed in, e.g.,~\cite{Jordan-1999,hafner2020action,liu2020deep,hou2014effective,hou2017effective,lentner2020,Anantharam-Borkar-07,krsav01,krsav07,Cuff-Permuter-Cover-IT10}.

The above discussion suggest that \ac{ILD} coding might be impossible. There is, however, one more option. Observe that one-to-many mappings are invertible if the images of any pair of inputs are disjoint, see Fig.~\ref{fig:o2m}. This opens the possibility to combine \ac{DM} and \ac{RC}/\ac{RNG} to create an \emph{invertible one-to-many} mapping. To ensure that the \ac{RC} is efficient, we add the requirement that
\begin{enumerate}
    \setcounter{enumi}{\numexpr4\relax}
    \item the \ac{RC} rate vanishes with $n$.
\end{enumerate}
Our main contribution is to construct invertible one-to-many encoders with rates approaching $\entrp{\qmf{A}}$, exponentially decaying \ac{I-divergence}, and vanishing \ac{RC} rate in $n$.

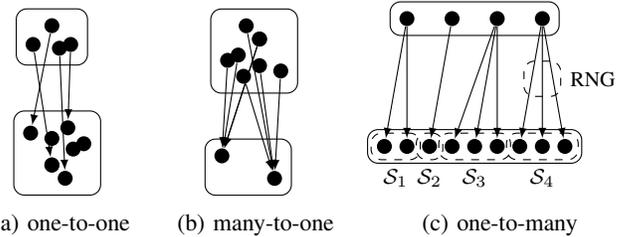
\begin{figure}
    \centering
    \begin{subfigure}[b]{0.25\columnwidth}
        \centering
        \begin{tikzpicture}[scale=1.0,rotate=-90]
\node(in1)[circle, fill,inner sep=2] at (0,0){};
\node(in2)[circle, fill,inner sep=2] at (0.3,0.1){};
\node(in3)[circle, fill,inner sep=2] at (0.25,-0.25){};
\node(in4)[circle, fill,inner sep=2] at (0.25,0.25){};

\node[draw,rounded corners,fit= (in1) (in2) (in3) (in4)]{};

\begin{scope}[xshift=1.5cm,rotate=45]
\node(out1)[circle, fill,inner sep=2] at (0,0){};
\node(out2)[circle, fill,inner sep=2] at (0.3,0.1){};
\node(out3)[circle, fill,inner sep=2] at (0.25,-0.25){};
\node(out4)[circle, fill,inner sep=2] at (0.05,0.25){};
\node(out5)[circle, fill,inner sep=2] at (0.25,-0.25){};
\node(out6)[circle, fill,inner sep=2] at (-0.25,-0.15){};
\node(out7)[circle, fill,inner sep=2] at (0.5,-0.25){};
\node(out8)[circle, fill,inner sep=2] at (0.35,0.25){};
\end{scope}

\node[draw,rounded corners,fit= (out1) (out8) (out7) (out4) (out6)]{};

\draw [-latex] (in1) -- (out6);
\draw [-latex] (in2) -- (out7);
\draw [-latex] (in3) -- (out5);
\draw [-latex] (in4) -- (out4);
\end{tikzpicture}
        %\vspace{1.3cm}
        \caption{one-to-one}
        \label{fig:o2o}
    \end{subfigure}
    ~ %add spacing between images, e.g. ~, \quad, \qquad, \hfill etc. %(or a blank line to force the subfigure onto a new line)
    \begin{subfigure}[b]{0.25\columnwidth}
        \centering
        \begin{tikzpicture}[scale=1.0,rotate=90]
\node(in1)[circle, fill,inner sep=2] at (0,-0.2){};
\node(in2)[circle, fill,inner sep=2] at (0.3,0.5){};
%\node(in3)[circle, fill,inner sep=2] at (0.0,0.4){};
%\node(in4)[circle, fill,inner sep=2] at (0.2,0.23){};
%\node(in5)[circle, fill,inner sep=2] at (0.22,0.15){};

\node[draw,rounded corners,fit= (in1) (in2) ]{};

\begin{scope}[xshift=1.5cm,rotate=45]
\node(out1)[circle, fill,inner sep=2] at (0,0){};
\node(out2)[circle, fill,inner sep=2] at (0.3,0.1){};
\node(out3)[circle, fill,inner sep=2] at (0.25,-0.25){};
\node(out4)[circle, fill,inner sep=2] at (0.05,0.25){};
\node(out5)[circle, fill,inner sep=2] at (0.25,-0.25){};
\node(out6)[circle, fill,inner sep=2] at (-0.25,-0.15){};
\node(out7)[circle, fill,inner sep=2] at (0.5,-0.25){};
\node(out8)[circle, fill,inner sep=2] at (0.35,0.25){};
\end{scope}

\node[draw,rounded corners,fit= (out1) (out8) (out7) (out4) (out6)]{};

\draw [latex-] (in1) -- (out1);
\draw [latex-] (in2) -- (out2);
\draw [latex-] (in1) -- (out4);
\draw [latex-] (in2) -- (out3);
\draw [latex-] (in1) -- (out6);
\draw [latex-] (in2) -- (out5);
\draw [latex-] (in1) -- (out7);
\draw [latex-] (in2) -- (out8);
\end{tikzpicture}
        %\vspace{1.3cm}
        \caption{many-to-one}
        \label{fig:m2o}
    \end{subfigure}
    ~ %add spacing between images, e.g. ~, \quad, \qquad, \hfill etc. %(or a blank line to force the subfigure onto a new line)
    \begin{subfigure}[b]{0.4\columnwidth}
        \centering
        \begin{tikzpicture}[scale=1,rotate=-90,yscale=-1]
\foreach \i in {1,...,4}{
\pgfmathsetmacro{\x}{\i*0.6}
\node(in\i)[circle, fill,inner sep=2] at (0.3,\x){};
}

%\node(in1)[circle, fill,inner sep=2] at (0,0.6){};
%\node(in2)[circle, fill,inner sep=2] at (-0.2,0.7){};
%\node(in3)[circle, fill,inner sep=2] at (0,1.0){};
%\node(in4)[circle, fill,inner sep=2] at (-0.3,1.1){};

\foreach \i in {1,...,9}{
\pgfmathsetmacro{\x}{\i*0.3}
\node(out\i)[circle, fill,inner sep=2] at (2,\x){};
}

%\node(out1)[circle, fill,inner sep=2] at (2,0){};
%\node(out2)[circle, fill,inner sep=2] at (2.2,0.15){};
%\node(out3)[circle, fill,inner sep=2] at (2.0,0.3){};

%\node(out4)[circle, fill,inner sep=2] at (2,0.6){};
%\node(out5)[circle, fill,inner sep=2] at (2.2,0.75){};
%\node(out6)[circle, fill,inner sep=2] at (2.0,0.9){};

%\node(out7)[circle, fill,inner sep=2] at (2,1.2){};

%\node(out8)[circle, fill,inner sep=2] at (2.2,1.6){};
%\node(out9)[circle, fill,inner sep=2] at (2.0,1.6){};

\foreach [count=\c] \i in {{1,2,3},{4,5,6},{7},{8,9}}{
	\foreach \j in \i{
	\draw[-latex] (in\c) -- (out\j);
	}
}

\node[draw,dashed,rounded corners,fit= (out1)(out2) (out3),inner sep=2,label={[]-90:\footnotesize$\mathcal{S}_4$}]{};
\node[draw,dashed,rounded corners,fit= (out4) (out5) (out6),inner sep=2,label={[]-90:\footnotesize$\mathcal{S}_3$}]{};
\node[draw,dashed,rounded corners,fit= (out7),inner sep=2 ,label={[]-90:\footnotesize$\mathcal{S}_2$}]{};
\node[draw,dashed,rounded corners,fit= (out8) (out9),inner sep=2,label={[]-90:\footnotesize$\mathcal{S}_1$}]{};

\node[draw,dashed,rounded corners,label={[]0:\footnotesize RNG}] at (1.1,0.6){\phantom{A}};

\node[draw,rounded corners,fit= (in1) (in4) (in2) (in3)]{};
\node[draw,rounded corners,fit= (out1) (out9)]{};

\end{tikzpicture}
        \caption{one-to-many}
        \label{fig:o2m}
    \end{subfigure}
    \caption{(a) Distribution matching (DM); (b) resolution coding (RC); and (c) invertible low-divergence (ILD) coding.}
    \label{fig:animals}
\end{figure}

This paper is organized as follows. 
Sec.~\ref{sec:preliminaries} introduces notation and bounds and
Sec.~\ref{sec:model} specifies the model and requirements. 
Sec.~\ref{sec:encoder-design} develops an encoder with a one-to-many mapping and a \ac{RC}. Sec.~\ref{sec:DM} treats \ac{DM} and generalizes results of~\cite{schulte2017divergence}. 
Sec.~\ref{sec:MLF-and-LLF} introduces the \ac{MLF} and \ac{LLF} algorithms for encoder design. Sec.~\ref{sec:I-div-lower-bounds} develops lower bounds on the \ac{I-divergence}. 
Sec.~\ref{sec:numerical-results} provides numerical results and compares them to the bounds.
Sec.~\ref{sec:conclusions-outlook} concludes the paper. 
Appendixes~\ref{app:inequalities}-\ref{app:MLF-LLF-general-alphabets} provide proofs of Lemmas and Theorems.

%-------------------------
\section{Preliminaries}
\label{sec:preliminaries}
\subsection{Notation}
Sets are written with calligraphic letters $\mset{A}$ and the empty set with $\emptyset$.
The cardinality of $\mset{A}$ is $|\mset{A}|$ and the $n$-fold Cartesian product of $\mset{A}$ is $\mset{A}^n$. 

Random variables (RVs) \acused{RV}are written with uppercase letters such as  $\rv{A}$, their realizations with corresponding lowercase letters $a$, and their alphabets as $\mset{A}$. A \ac{pmf} of a \ac{RV} \rv{A} is denoted by $\pmf{A}$ or $\qmf{A}$.
We use $\qmf{A}$ for target \acp{pmf} and $\pmf{A}$ for synthesized \acp{pmf}. We discard the subscripts when referring to generic \acp{pmf}. A \ac{pmf} or function is sometimes written as a vector, e.g., \ac{pmf} $\qmf{A}$ with alphabet $\mset{A}=\{1,\ldots,|\mset{A}|\}$ is written as $[Q_\rv{A}(1),\ldots,Q_\rv{A}(|\mset{A}|)]$. The uniform \ac{pmf} over a set of $\numW$ elements is denoted by $U_\numW$.

A random string is denoted by $\rv{A}^n = \rv{A}_1\rv{A}_2\ldots\rv{A}_n$ and its realizations by $a^n = a_1a_2\ldots a_n \in \mset{A}^n$. We write $Q_{\rv{A}^n}=\qmf{A}^n$ for the \ac{pmf} of a string of \ac{iid} \acp{RV}. The probability of a set $\mset{S}\subseteq\mset{A}^n$ of strings with respect to $\qmf{A}^n$ is written as
\begin{equation}
    \qmf{A}^n(\mset{S}) := \sum_{a^n\in\mset{S}} \qmf{A}^n(a^n)
\end{equation}
and as $q_{\mset{S}}=\qmf{A}^n(\mset{S})$ for short. Probabilities conditioned on the event $\mset{S}$ are written as
\begin{equation}
    Q_{\rv{A}|\mset{S}}^n(a^n) := \begin{cases} \frac{\qmf{A}^n(a^n)}{\qmf{A}^n(\mset{S})},& a^n\in\mathcal{S}\\
    0,& a^n\not\in \mathcal{S} .
    \end{cases}
    \label{eq:Qset}
\end{equation}

Let $n_i=n_i(a^n)$ be the number of occurrences of the letter $i$ in $a^n$ for $i=1,\dots,|\mset{A}|$. The \emph{empirical} \ac{pmf} (or \emph{type}) of $a^n$ is $\pi_{a^n}=\frac1n[n_1,\dots,n_{|\mset{A}|}]$. Let $\mset{P}_n$ be the set of empirical \acp{pmf} with denominator $n$ (the \emph{$n$-types}). The string $a^n$ is called \emph{typical} with respect to $P$ and $\epsilon$ if (see~\cite[Ch.~2.4]{ElGamal2011})
\begin{align}
    \left| \pi_{a^n}(i) - P(i)\right| \le \epsilon\, P(i)
    \label{eq:typical}
\end{align}
for all $i\in\mset{A}$. The set of typical strings is denoted $\mset{T}_\epsilon(P)$.

The \emph{expectation} of a real-valued function $f$ of a random variable $\rv{A}$ with respect to $P$ is
\begin{align}
  \E{P}{f(\rv{A})} = \sum\limits_{a\in\supp(P)} P(a) f(a)
\end{align}
where $\supp(P) \subseteq \mset{A}$ is the support of $P$, i.e., the set of $a\in\mset{A}$ with $P(a)>0$. For example, the \emph{variance} of $f(A)$ is
\begin{align}
  \Var{P}{f(A)}
  = \E{P}{f(\rv{A})^2} - \E{P}{f(\rv{A})}^2.
\end{align}

The \emph{self-information} of $a$ with respect to a \ac{pmf} $P$ is
\begin{align}
    \selfinfo{P}{a} = - \log_2 P(a).
    \label{eq:self-info}
\end{align}
The \emph{entropy} of $P$ is 
\begin{align}
  \entrp{P}
  = \E{P}{\selfinfo{P}{\rv{A}}}
  = \E{P}{- \log_2 P(\rv{A})}.
\end{align}
The binary entropy function is $h(p)=-p\log_2(p)-(1-p)\log_2(1-p)$ for $0<p<1$ and $h(p)=0$ otherwise. The average conditional entropy is written as
\begin{align}
  \entrp{P_{\rv{A}|\rv{W}}} = \sum_{w\in\supp(\pmf{W})} \pmf{W}(w)\, \entrp{P_{\rv{A}|\rv{W}}(\cdot|w)}.
\end{align}

The \emph{cross entropy} of two \ac{pmf}s $P$ and $Q$ is
\begin{align}
  \crossentrp{P}{Q}
  = \E{P}{- \log_2 Q(\rv{A})}.
\end{align}
For example, we have $\selfinfo{\qmf{A}^n}{a^n} = n \crossentrp{\pi_{a^n}}{\qmf{A}}$.
The \ac{I-divergence} of two \ac{pmf}s $P$ and $Q$ is
\begin{align}
  \diverg{P}{Q}
  = \E{P}{\log_2 \frac{P(\rv{A})}{Q(\rv{A})}}
\end{align}
and we have
\begin{align}
  \crossentrp{P}{Q}=\entrp{P}+\diverg{P}{Q}.
  \label{eq:X-H-D}
\end{align}
\ac{I-divergence} is also known as relative entropy and Kullback-Leibler divergence~\cite[Ch.~2.3]{cover2006elements}.

The $\ell_1$ distance between two \ac{pmf}s $P$ and $Q$ on $\mset{A}$ is
\begin{align}
    \done{P}{Q} = \sum_{a \in \mset{A}} \left| P(a)-Q(a) \right|
\end{align}
where $\done{P}{Q}\le2$ with equality if and only if the supports of $P$ and $Q$ are disjoint. For sequences $f(n)$ and $g(n)$, $n=1,2,\dots$, the little-o notation $f(n)=o(g(n))$ means that $\lim_{n\rightarrow\infty} f(n)/g(n)\rightarrow0$, see~\cite[p.~61]{Landau1909}.

%-------------------------
\subsection{Bounds for Entropy and \ac{I-divergence}}
\label{subsec:bounds-entropy-Idiv}

We state several results that we need below. The shorthand $d_1$ refers to $\done{P}{Q}$.

\begin{lemma}[{\cite[Ch.~2.6]{cover2006elements}}]
\label{lemma:divergence-bound}
$\diverg{P}{Q}\ge 0$ and $\crossentrp{P}{Q}\ge\entrp{P}$, both with equality if and only if $P=Q$.
\end{lemma}

\begin{lemma}[{\cite[Ch.~2.7]{cover2006elements}}]
\label{lemma:divergence-convexity}
$\diverg{P}{Q}$ is convex in the \ac{pmf}-pair $(P,Q)$. $\crossentrp{P}{Q}$ is linear in $P$ and convex in $Q$. $\entrp{P}$ is concave in $P$.
\end{lemma}

\begin{lemma}[{\cite[Sec.~17.3]{cover2006elements}}]
\label{lemma:entropy-inequality}
If $d_1\le1/2$ then
\begin{align}
    \left| \entrp{P} - \entrp{Q} \right| \le -d_1 \log_2\frac{d_1}{|\mset{A}|}.
    \label{eq:entropy-inequality}
\end{align}
\end{lemma}

\begin{lemma}
\label{lemma:entropy-inequality2}
Let $p_{\rm min}$ and $p_{\rm max}$ be the respective minimum and maximum probabilities of $P$. If $d_1\le2 p_{\rm min}$ then
\begin{align}
    \entrp{P} - \entrp{Q} \le \frac{d_1}{2} \log_2\frac{p_{\rm max}+d_1/2}{p_{\rm min}-d_1/2}.
    \label{eq:entropy-inequality2}
\end{align}
\end{lemma}
\begin{proof}
See Appendix~\ref{app:inequalities} and~\cite[Lemma~1]{schulte2017divergence}.
\end{proof}

\begin{lemma}[Pinsker Inequalities {\cite[Ch.~11]{cover2006elements}, \cite[Eq.~(23)]{Sason-Verdu-IT16}}]
\label{lemma:pinsker}
\begin{align}
  \frac{1}{2\ln2} \, d_1^2
  \le \diverg{P}{Q} \le
  \frac{1}{q_{\rm min} \ln2} \, d_1^2
  \label{eq:pinsker-inequality}
\end{align}
where $q_{\rm min}=\min_{a\in\supp(P)} Q(a)$. For instance, the right-hand side of \eqref{eq:pinsker-inequality} is $\infty$ if $Q(a)=0$ when $P(a)>0$.
\end{lemma}

We remark that one advantage of Lemma \ref{lemma:entropy-inequality2} over Lemmas~\ref{lemma:entropy-inequality} and~\ref{lemma:pinsker} is that its bound is effectively linear in $d_1$ for small $d_1$ rather than behaving as $-d_1 \log d_1$ or $d_1^2$.

%-------------------------
\subsection{Probability Bounds for Sums and Sets}
\label{subsec:bounds-sums}

The following Lemma specializes a result of~\cite{Hoeffding1963} to \ac{iid} discrete random variables.

\begin{lemma}[Hoeffding {\cite[Thm.~1]{Hoeffding1963}}]
\label{lemma:hoeffding}
Consider the \ac{iid} string $\rv{A}^n$ of random variables where $\rv{A}:=\rv{A}_1$ has \ac{pmf} $P$ with alphabet $\mset{A}$. Let $\rv{S}_n=\frac1n\sum_{i=1}^n f(\rv{A}_i)$ for a real-valued function $f$ satisfying $0\le f(a)\le 1$ for all $a\in\mset{A}$. We have
\begin{align}
  \Pr\left[ \rv{S}_n - \E{P}{f(\rv{A})} \ge t \right] \le e^{-2nt^2}
  \label{eq:hoeffding-inequality}
\end{align}
for $t\ge0$.
\end{lemma}

The next lemma gives basic bounds for typical strings.

\begin{lemma}[{\cite[Ch.~2.4]{ElGamal2011}}]
\label{lemma:typical-sets}
Consider the \ac{pmf} $P$ with alphabet $\mset{A}$. Let $a^n\in\mset{T}_\epsilon(P)$ and let $p_{\rm min}=\min_{a\in\supp(\mset{A})} P(a)$. We have
\begin{align}
  & 1 - \delta_{\epsilon} \le P\left( {\mset{T}_\epsilon(P)} \right) \le 1
  \label{eq:typical-sets-1} \\
  & 2^{ - n \entrp{P} (1+\epsilon)} \le P^n(a^n) \le 2^{-n \entrp{P} (1-\epsilon)}
  \label{eq:typical-sets-2} \\
  & (1 - \delta_{\epsilon}) 2^{n \entrp{P} (1-\epsilon)} \le \left| \mset{T}_\epsilon(P) \right| \le 2^{n \entrp{P} (1+\epsilon)}
  \label{eq:typical-sets-3}
\end{align}
where $\delta_{\epsilon} = 2 |\mset{A}| \exp\left(-2n \, p_{\rm min}^2\, \epsilon^2\right)$.
\end{lemma}
\begin{proof}
The left-hand side of \eqref{eq:typical-sets-1} follows by Lemma~\ref{lemma:hoeffding} and the union bound, the bounds \eqref{eq:typical-sets-2} by the definition of typical strings, and the bounds \eqref{eq:typical-sets-3} by \eqref{eq:typical-sets-1} and \eqref{eq:typical-sets-2}.
\end{proof}

%-------------------------
\subsection{Bounds for Binomial and Multinomial Coefficients}
We state several results for binomial coefficients.

\begin{lemma}[see {\cite[p.~166]{graham-etal-89}}]
\label{lemma:binomial-identity}
For a non-negative integer $k$ and a positive integer $n$ with $k\le n$ we have
\begin{align}
    \sum_{i=0}^k \binom{n}{i} \left(\frac{n}{2} - i \right) = \frac{k+1}{2} \binom{n}{k+1}
    = \frac{n-k}{2} \binom{n}{k}.
    \label{eq:binomial-identity}
\end{align}
\end{lemma}

\begin{lemma} \label{lemma:binomial-bound1}
For $0<p=k/n<1$ we have~\cite[p.~530]{gallager1968information}
\begin{align}
    \frac{2^{nh(p)}}{\sqrt{8 n p (1-p)}}
    \le \binom{n}{np} \le \frac{2^{nh(p)}}{\sqrt{2 \pi n p (1-p)}}.
    \label{eq:binomial-bound1}
\end{align}
\end{lemma}

\begin{lemma} \label{lemma:binomial-bound2}
For $0\le p=k/n<1/2$ we have~\cite[Eq.~(25)]{bahadur1960} (see also~\cite[Lemma~3]{schulte2017divergence})
\begin{align}
    \alpha \beta \binom{n}{np} \le
    \sum_{i=0}^{np} \binom{n}{i} \le \alpha \binom{n}{np}
    \label{eq:binomial-bound2}
\end{align}
where
\begin{align}
  \alpha = \frac{1-p+1/n}{1-2p+1/n}, \quad \beta = \frac{(1-2p)^2}{(1-2p)^2+1/n}.  
\end{align}
\end{lemma}

Multinomial coefficients are written as
\begin{align}
    \binom{n}{n_1\ldots n_{|\mset{A}|}} = \frac{n!}{\prod_{i=1}^{|\mset{A}|} n_i!}
    \label{eq:multinomial}
\end{align}
with the integers $0 \le n_i\le n$ and $n=\sum_{i=1}^{|\mset{A}|}n_i$. An analog of Lemma~\ref{lemma:binomial-bound1} is as follows.

\begin{lemma} \label{lemma:multinomial-bound1}
For the \ac{pmf} $P=[p_1,\dots,p_{|\mset{A}|}]$ and $0<p_i=n_i/n<1$ for all $i$, we have
\begin{align}
    & \frac{2^{n\entrp{P}}}
    {\left[(8n)^{|\mset{A}|-1} \prod_{i=1}^{|\mset{A}|}p_i \right]^{1/2}} 
    \le \binom{n}{np_1\ldots np_{|\mset{A}|}} \nonumber \\
    & \qquad \le \frac{2^{n\entrp{P}}}{\left[ (2\pi n)^{|\mset{A}|-1} \prod_{i=1}^{|\mset{A}|}p_i \right]^{1/2}}.
    \label{eq:multinomial-bound1}
\end{align}
\end{lemma}
\begin{proof}
Use the binomial expansion
\begin{align}
    \binom{n}{np_1\ldots np_{|\mset{A}|}} =
    \prod_{i=1}^{|\mset{A}|-1}
    \binom{n\left( 1-\sum_{j=1}^{i-1} p_j \right)}{n p_i}
    \label{eq:multinomial2}
\end{align}
and apply the bounds \eqref{eq:binomial-bound1} to \eqref{eq:multinomial2}.
\end{proof}

%-------------------------
\section{Model and Requirements}
\label{sec:model}

Consider the model depicted in Fig.~\ref{fig:modelDMpluResolution}.
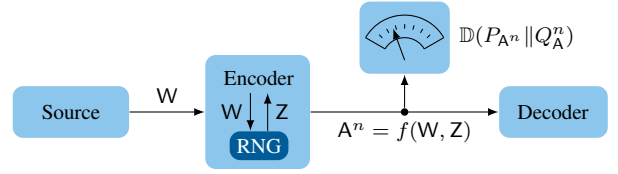
\begin{figure}[t]
    \centering
    %\begin{tikzpicture}[font=\footnotesize]
%\node[rounded corners,fill=TUMAkzBlue1](encoder) at (0,0){
%\begin{tikzpicture}
%\node(enc) at (0,0){Encoder};
%\path(enc) --++(0,-0.8) node[rounded corners,fill=TUMSecBlue1,text=white] (rng){RNG};
%\draw[-latex] (rng.60) -- (enc.south -| rng.60);
%\draw[-latex]  (enc.south -| rng.120)--(rng.120);
%\end{tikzpicture}
%};
%\path (encoder.west) --++ (-1,0)node[rounded corners,fill=TUMAkzBlue1,anchor=east,text width=30,text centered, inner sep=7](data){$U_\numW$};
%%\path (encoder.east) --++ (1,0)node[rounded corners,fill=TUMAkzBlue1,anchor=west,inner sep=5](measurement){$\mathbb{D}(\cdot\Vert Q_\rv{A}^n)$};
%\path (encoder.east) --++ (1,0)node[anchor=west,inner sep=0,label={[above]:$\mathbb{D}(\cdot\Vert Q_\rv{A}^n)$}](measurement){
%	\begin{tikzpicture}[scale=0.6]
%	\fill[rounded corners,fill=TUMAkzBlue1] (-1,-0.3) rectangle (1,1.3);
%	\draw (30:1) arc (30:150:1) ;
%	\draw (30:0.5) arc (30:150:0.5) ;
%	\draw  (30:1) -- (30:0.5);
%	\draw  (150:1) -- (150:0.5);
%	\draw[-latex](0,0) -- (110:0.75);
%	\foreach \degr in {45,60,...,135}{
%	\draw (\degr:0.7)--(\degr:0.8);}
%	\end{tikzpicture}
%};
%\path (measurement.east) --++ (1,0)node[rounded corners,fill=TUMAkzBlue1,anchor=west,inner sep=7](decoder){Decoder};

%\draw[-latex] (data)--node[above]{$\rv{W}$}(encoder);
%\draw[-latex] (encoder)--node[above]{$\rv{A}^n$}(measurement);
%\draw[-latex] (measurement)--node[above]{$\rv{A}^n$}(decoder);
%\end{tikzpicture}

\begin{tikzpicture}[font=\footnotesize]
\node[rounded corners,fill=TUMAkzBlue1](encoder) at (0,0){
\begin{tikzpicture}
\node(enc) at (0,0){Encoder};
\path(enc) --++(0,-0.9) node[rounded corners,fill=TUMSecBlue1,text=white] (rng){RNG};
\draw[-latex] (rng.60) -- node[right]{$\rv{Z}$}(enc.south -| rng.60);
\draw[-latex]  (enc.south -| rng.120) -- node[left]{$\rv{W}$}(rng.120);
\end{tikzpicture}
};
\path (encoder.west) --++ (-1,0)node[rounded corners,fill=TUMAkzBlue1,anchor=east,text width=30,text centered, inner sep=7](data){Source};
%\path (encoder.east) --++ (1,0)node[rounded corners,fill=TUMAkzBlue1,anchor=west,inner sep=5](measurement){$\mathbb{D}(\cdot\Vert Q_\rv{A}^n)$};
\path (encoder.east) --++ (1.25,0.5)node[anchor=south,inner sep=0,label={0:$\mathbb{D}(P_{\rv{A}^n}\Vert Q_\rv{A}^n)$}](measurement){
	\begin{tikzpicture}[scale=0.6]
	\fill[rounded corners,fill=TUMAkzBlue1] (-1,-0.3) rectangle (1,1.3);
	\draw (30:1) arc (30:150:1) ;
	\draw (30:0.5) arc (30:150:0.5) ;
	\draw  (30:1) -- (30:0.5);
	\draw  (150:1) -- (150:0.5);
	\draw[-latex](0,0) -- (110:0.75);
	\foreach \degr in {45,60,...,135}{
	\draw (\degr:0.7)--(\degr:0.8);}
	\end{tikzpicture}
};
\path (encoder.east) --++ (2.5,0)node[rounded corners,fill=TUMAkzBlue1,anchor=west,inner sep=7](decoder){Decoder};

\draw[-latex] (data)--node[above]{$\rv{W}$}(encoder);
\draw[-latex] (encoder)--node[below]{$\rv{A}^n=f(\rv{W},\rv{Z})$}  node(wiretap)[fill,circle,inner sep=1pt]{} (decoder);
\draw[-latex] (wiretap) -- (measurement);
\end{tikzpicture}
    \caption{Transmission experiment.} 
    \label{fig:modelDMpluResolution}
\end{figure}
The source generates a message $\rv{W}$ with \ac{pmf} $\pmf{W}=U_K$.
The information rate is
\begin{equation}
    \Rinfo = \frac{\entrp{\pmf{W}}}{n} = \frac{\log_2 \numW }{\bl}.
\end{equation}
To permit randomization, the encoder is given a \ac{RNG} that generates an index $\rv{Z}(w)$ with \ac{pmf} $P_{\rv{Z}|\rv{W}}(\cdot|w)$ given $\rv{W}=w$. For example, one may choose $\rv{Z}(w)=\rv{A}^n(w)$. We consider two types of \acp{RNG}, namely idealized \acp{RNG} and \acp{RNG} based on \acp{RC}. One can measure the \ac{RNG} rates in two ways: with the average conditional entropy $\entrp{P_{\rv{Z}|\rv{W}}}$ and with the number $\Brng$ of \ac{RC} bits. The resulting rates are
\begin{align}
    \Hrng = \frac{\entrp{P_{\rv{Z}|\rv{W}}}}{\bl}, \quad
    \Rrng = \frac{\Brng}{\bl}
\end{align}
and $\Hrng\le\Rrng$ because $\rv{Z}(w)$ is a function of the \ac{RC} bits for all $w$.

The encoder output is $\rv{A}^n=f(\rv{W},\rv{Z})$ for some function $f$. The resolution quality is measured via the \ac{I-divergence}
\begin{align}
\diverg{P_{\rv{A}^n}}{Q^n_{\rv{A}}}
\label{eq:I-div}
\end{align}
for a specified \ac{pmf} $\qmf{A}$. For instance, for the stealth problem a warden knows the target \ac{pmf} $Q_\rv{A}$, the code statistics $P_{\rv{A}^n}$, and the \ac{RNG} statistics $P_{\rv{Z}|\rv{W}}$. Given $a^n$ the warden must decide whether a code word was transmitted or not. One can show~\cite{hou2014effective,hou2017effective} that the best the warden can do is to guess if \eqref{eq:I-div} vanishes with the block length $n$.

The problem requirements are thus as follows: the decoder must recover $\rv{W}$ without error, $\diverg{P_{\rv{A}^n}}{Q^n_{\rv{A}}}$ must vanish with growing $n$, and $\Rrng$ must vanish with growing $n$. The rate $\Rinfo$ is said to be achievable if these requirements are met. We wish to maximize the achievable rate. In fact, these requirements are coupled, as shown below. For example, vanishing $\frac1n \diverg{P_{\rv{A}^n}}{Q^n_{\rv{A}}}$ and $\Rrng$ imply that $\Rinfo\rightarrow\entrp{\qmf{A}}$.

Observe that if $\qmf{A}=U_{|\mset{A}'|}$ for any $\mset{A}'\subseteq\mset{A}$ then one achieves zero \ac{I-divergence} at maximal rate $\Rinfo=\log_2 |\mset{A}'|$ without a \ac{RNG} by choosing $K=|\mset{A}'|^n$ and putting out the $|\mset{A}'|$-ary representation of $w-1$. We hence focus on $\qmf{A}$ that are not uniformly distributed over any subset.

%-------------------------
\subsection{Rate Bounds}
We use the bounding approach of~\cite[Sec.~1.3.3]{hou2017effective}. The linearity of cross entropy gives
\begin{align}
  \crossentrp{P_{\rv{A}^n}}{\qmf{A}^n}
  = n \, \crossentrp{\bar{P}_{\rv{A}}}{\qmf{A}}
  \label{eq:crossentrp-identity}
\end{align}
where $\bar{P}_{\rv{A}}=\frac1n\sum_{i=1}^n P_{\rv{A}_i}$ is the average letter \ac{pmf} of $A^n$. Lemma \ref{lemma:divergence-convexity} and
\eqref{eq:X-H-D} further give
\begin{align}
  \entrp{P_{\rv{A}^n}}
  & \le \sum_{i=1}^n H(P_{\rv{A}_i})
  \le n\,\entrp{\bar{P}_{\rv{A}}} \label{eq:entropy-concavity-bound} \\
  \diverg{P_{\rv{A}^n}}{\qmf{A}^n}
  & = n\,\crossentrp{\bar{P}_{\rv{A}}}{\qmf{A}} - \entrp{P_{\rv{A}^n}} \nonumber \\
  & \ge n\,\diverg{\bar{P}_{\rv{A}}}{\qmf{A}} .
  \label{eq:div-convexity-bound}
\end{align}
We have the following lemmas.

\begin{lemma} \label{lemma:rate-upper}
Vanishing $\frac1n \diverg{P_{\rv{A}^n}}{Q^n_{\rv{A}}}$ requires
\begin{align}
   \Rinfo \le \entrp{\qmf{A}}.
   \label{eq:Rinfo-HA}
\end{align}
Moreover, if the decoder can recover both the message $\rv{W}$ and the \ac{RNG} index $\rv{Z}$, then we have the stronger bound
\begin{align}
   \Rinfo + \Hrng \le \entrp{\qmf{A}}.
   \label{eq:Rinfo+Rrng-upper}
\end{align}
\end{lemma}
\begin{proof}
Consider $0\le\xi\le\frac12$ and
\begin{align}
  \frac1n \diverg{P_{\rv{A}^n}}{Q^n_{\rv{A}}}
  \le \frac{\xi^2}{2\ln 2} .
  \label{eq:norm-div-bound}
\end{align}
The bound \eqref{eq:div-convexity-bound} and Lemmas \ref{lemma:entropy-inequality} and \ref{lemma:pinsker} give
\begin{align}
   \left| \entrp{\bar{P}_{\rv{A}}} - \entrp{\qmf{A}} \right| \le - \xi\, \log_2 \frac{\xi}{|\mset{A}|}
   \label{eq:Rinfo-d1-bounds}
\end{align}
and therefore
\begin{align}
   \entrp{\bar{P}_{\rv{A}}} \le \entrp{\qmf{A}} - \xi\, \log_2 \frac{\xi}{|\mset{A}|}.
   \label{eq:Rinfo-upper}
\end{align}
We further have $\entrp{\pmf{W}}\le\entrp{P_{\rv{A}^n}}$ since $\rv{W}$ is a function of $\rv{A}^n$. Combining this bound with \eqref{eq:entropy-concavity-bound} and \eqref{eq:Rinfo-upper} proves \eqref{eq:Rinfo-HA} for $\xi\rightarrow0$. To prove \eqref{eq:Rinfo+Rrng-upper}, note that if $\rv{W}$ and $\rv{Z}$ are functions of $\rv{A}^n$ then $\entrp{\pmf{WZ}}\le \entrp{P_{\rv{A}^n}}$.
\end{proof}

A reverse bound to \eqref{eq:Rinfo+Rrng-upper} holds more generally. 

\begin{lemma} \label{lemma:rate-lower}
Vanishing $\frac1n \diverg{P_{\rv{A}^n}}{Q^n_{\rv{A}}}$ requires
\begin{align}
  \Rinfo + \Hrng \ge \entrp{\qmf{A}}.
  \label{eq:Rinfo+Rrng-lower}
\end{align}
\end{lemma}
\begin{proof}
Consider the bound \eqref{eq:norm-div-bound}. We have
\begin{align}
   \Rinfo + \Hrng
%   & \overset{(a)}{\ge} \frac{1}{n} \entrp{P_{\rv{A}^n}} \nonumber \\
   & \overset{(a)}{\ge} \frac{1}{n} \entrp{P_{\rv{A}^n}} + \left(\frac1n \diverg{P_{\rv{A}^n}}{\qmf{A}^n} - \frac{\xi^2}{2\ln 2} \right) \nonumber \\
   & \overset{(b)}{=} \crossentrp{\bar{P}_{\rv{A}}}{\qmf{A}} - \frac{\xi^2}{2\ln 2} \nonumber \\
   & \overset{(c)}{\ge} \entrp{\bar{P}_{\rv{A}}} - \frac{\xi^2}{2\ln 2} \nonumber \\
   & \overset{(d)}{\ge}
   \entrp{\qmf{A}} + \xi\, \log_2 \frac{\xi}{|\mset{A}|} - \frac{\xi^2}{2\ln 2}
   \label{eq:Rinfo-lower}
\end{align}
where $(a)$ follows because $\rv{A}^n$ is a function of $\rv{W}$ and $\rv{Z}$ and by hypothesis \eqref{eq:norm-div-bound}, $(b)$ follows by~\eqref{eq:X-H-D} and~\eqref{eq:crossentrp-identity}, $(c)$ follows by Lemma~\ref{lemma:divergence-bound}, and $(d)$ follows by \eqref{eq:Rinfo-d1-bounds}. Finally, let $\xi\rightarrow0$.
\end{proof}

Lemma~\ref{lemma:rate-lower} is valid for \ac{DM}, \ac{RC}, and for one-to-many mappings. For example, if $\Rinfo\rightarrow 0$ as $n\rightarrow\infty$ then we asymptotically require $\Rrng\ge\Hrng\ge\entrp{\qmf{A}}$. Finally, we remark that the inequalities~\eqref{eq:Rinfo+Rrng-upper} and~\eqref{eq:Rinfo+Rrng-lower} are usually strict for finite $n$ and hence it is not clear whether \ac{ILD} coding is possible.

%-------------------------
\subsection{Discussion}
\label{subsec:model-discussion}
Sec.~\ref{sec:introduction} reviews two approaches to approximate target \acp{pmf}, namely \ac{DM} and \ac{RC}. \ac{DM} uses a one-to-one mapping which is a special case of the above model without a \ac{RNG}. Vanishing normalized (or un-normalized) \ac{I-divergence} thus implies that $\Rinfo$ is asymptotically \emph{upper} bounded by $\entrp{Q_{\rv{A}}}$, see \eqref{eq:Rinfo-HA}. In fact, the \ac{I-divergence} of the best binary \ac{DM} grows as $\frac12 \log_2 n$ with $n$~\cite{schulte2017divergence}. Applications of \ac{DM}, such as probabilistic shaping for energy-efficient communication, usually require only vanishing \emph{normalized} \ac{I-divergence}. Algorithms for \ac{DM} that have $\Rinfo\rightarrow\entrp{\qmf{A}}$ for large $n$ include \ac{CCDM} \cite{bocherer2015bandwidth,schulte2016constant} and shell mapping \cite{laroia1994optimal,khandani1993shaping,gultekin2017constellation,schulte2019divergence}.

\ac{RC} uses a many-to-one mapping and the \ac{RC} rate for vanishing \ac{I-divergence} is asymptotically \emph{lower} bounded by $\entrp{Q_{\rv{A}}}$, see~\eqref{eq:Rinfo-lower}. To approach the lower bound, one can, e.g., apply random coding arguments~\cite{han1993approximation,Ha01}, interval algorithms~\cite{han1997interval}, fixed-to-variable length codes~\cite{bocherer2013fixed}, variable-to-fixed length codes~\cite{Amjad-RED16}, or fixed-to-fixed length codes~\cite{boecherer-geiger-IT16,boecherer2013block} such as polar codes~\cite{Blasco-Serrano-etal-IZS12,Chou-etal-ISIT15}. These algorithms use deterministic many-to-one mappings that are not invertible in general.

\ac{ILD} coding uses a one-to-many mapping that combines \ac{DM} and \ac{RC}, see Fig.~\ref{fig:o2m}. This is similar to randomized encoding which is a common tool in multi-user information theory, e.g., for \ac{RC}/\ac{RNG} and wiretap channels~\cite{wyner1975wiretap}. The differences between the approaches are subtle. In particular, we require zero error while the randomization for wiretap and other problems permits small error. Also, we must carefully design the \ac{DM} encoder and \ac{RC} because $\diverg{P_{\rv{A}^n}}{\qmf{A}^n}$ should vanish.

%------------------------
\section{Encoder Design}
\label{sec:encoder-design}

An \ac{ILD} encoder is a one-to-many mapping into disjoint sets, see Fig.~\ref{fig:o2m}. All strings $a^n$ assigned to message $w$ are collected in the set $\mset{S}_w$ and we require $\binset{v}\cap\binset{w} = \emptyset$ for $ v\ne w$. We denote the set of all strings under consideration as
\begin{align}
  \mset{S} = \bigcup_w \mset{S}_w
  \label{eq:code-partition}
\end{align}
where $\{\mset{S}_w : w=1,\dots,K \}$ partitions $\mset{S}$. A basic choice for $\mset{S}$ is $\supp(\ptarget)^n$. 

%-----------------------------
\subsection{Two-Step Encoding}

Encoding involves two-steps. First, the message $w$ chooses the set $\binset{w}$.
The encoder then requests an index $\rv{Z}(w)=\rv{A}^n(w)$ from the \ac{RNG} $P_{\rv{Z}|\rv{W}}(\cdot|w)$ that uniquely identifies a string $a^n=f(w,z)$ from $\binset{w}$. For this $a^n$, we have
\begin{align}
  P_{\rv{A}^n|\rv{W}}\left(a^n | w\right) = P_{\rv{Z}|\rv{W}}(z|w)
\end{align}
and for each $a^n\in\mset{S}_w$ we have
\begin{align}
  P_{\rv{A}^n}(a^n) = P_{\rv{A}^n W}(a^n,w) = \frac{1}{K} P_{\rv{A}^n|\rv{W}}\left(a^n|w \right).
  \label{eq:Pagw}
\end{align}

Suppose that $n$ and $Q_\rv{A}$ are given. Encoder design involves choosing the:
\begin{itemize}
    \item number $K$ of messages;
    \item code: the set $\mset{S}$ of strings;
    \item encoder map: sets $\binset{w}$, $w=1,\ldots,K$, that partition $\mset{S}$;
    \item RNG: \ac{pmf}s $P_{\rv{Z}|\rv{W}}(\cdot|w)$, $w=1,\ldots,K$.
\end{itemize}

The encoder output is $\rv{A}^n=f(\rv{W},\rv{Z})$ and $f$ is invertible, so we are in the case described for the bound \eqref{eq:Rinfo+Rrng-upper} and with
\begin{align}
    \entrp{P_{\rv{A}^n}}
    = \entrp{P_\rv{WZ}}
    = n\left( \Rinfo + \Hrng \right).
\end{align}
One might, therefore, consider the transmission rate to be $\Rinfo+\Hrng$ rather than $\Rinfo$. However, $\rv{Z}$ is non-uniform and generated by a many-to-one mapping in general so that one cannot necessarily recover the $\Brng=n\Rrng$ bits that generate $\rv{Z}$. Thus, we consider the information rate to be $\Rinfo$. At the same time, the encoder does ``share randomness'' via $\rv{Z}$.

%-----------------------------
\subsection{Idealized \ac{RNG}}

Recall that $q_{\mset{S}}=\qmf{A}^n(\mset{S})$ and  $Q^n_{\rv{A}|\mset{S}}(a^n)=\qmf{A}^n(a^n)/q_{\mset{S}}$ for $a^n\in\mset{S}$, see \eqref{eq:Qset}. We expand \eqref{eq:I-div} by using \eqref{eq:Pagw} as follows:
\begin{align}
    & \sum_{w=1}^K \, \sum_{a^n\in\binset{w}\cap\, \supp\left(P_{\rv{A}^n}\right)} 
    \frac{P_{\rv{A}^n|\rv{W}}\left(a^n|w \right)}{K}
    \log_2 \frac{\frac{1}{K} P_{\rv{A}^n|\rv{W}}\left(a^n|w \right)}{\qmf{A}^n(a^n)} \nonumber \\
    & = \diverg{U_\numW}{[q_{\binset{1}},\ldots,q_{\binset{\numW}}]} \nonumber \\
    & \qquad + \sum_{w=1}^K \frac{1}{\numW} \diverg{P_{\rv{A}^n|\rv{W}}(\cdot|w)}{ Q^n_{\rv{A}|\binset{w}}}. \label{eq:qusisplit}
\end{align}
The effects of the two-step encoding are apparent in~\eqref{eq:qusisplit}: the first term accounts for the choice of set $\mset{S}_w$ and the second term accounts for the \ac{RNG}. We will study the \acp{I-divergence}
\begin{align}
  & \diverg{U_\numW}{[q_{\binset{1}},\ldots,q_{\binset{\numW}}]} \label{eq:I-div1} \\
  & \diverg{P_{\rv{A}^n|\rv{W}}(\cdot|w)}{ Q^n_{\rv{A}|\binset{w}}} \label{eq:I-div2}
\end{align}
separately. The identity \eqref{eq:qusisplit} and Lemma~\ref{lemma:divergence-bound} give the following result.
\begin{proposition} \label{prop:best-RNG}
The encoder \ac{RNG} with
\begin{align}
    P_{\rv{A}^n|\rv{W}}(\cdot|w) = Q^n_{\rv{A}|\binset{w}}
    \label{eq:idealRNG}
\end{align}
for all $w$ gives the smallest \ac{I-divergence}
\begin{equation}
    \diverg{\pblock}{ Q^n_{\rv{A}}} =\diverg{U_\numW}{[q_{\mset{S}_1},\ldots,q_{\mset{S}_{\numW}}]}.
    \label{eq:qusi_simple}
\end{equation}
\end{proposition}

Proposition~\ref{prop:best-RNG} gives intuition on how to choose the partition $\{\mset{S}_1,\dots,\mset{S}_\numW\}$: the \ac{pmf} $[q_{\mset{S}_1},\dots,q_{\mset{S}_K}]$ should be close to uniform. Sec.~\ref{sec:MLF-and-LLF} develops algorithms that separate the $a^n$ into approximately equally likely sets with respect to $\qmf{A}^n$.

%-----------------------------
\subsection{\ac{RNG} via \ac{RC}}
\label{subsec:RNG-RC}

The idealized \ac{RNG} of \eqref{eq:idealRNG} cannot be implemented in general. To approximate it, various authors have developed theory and algorithms for \acp{RC} with vanishing \ac{I-divergence} \eqref{eq:I-div2}, see Sec.~\ref{subsec:model-discussion}. We study fixed-to-fixed length encoders generated by Algorithm 2 in~\cite{boecherer-geiger-IT16}. Consider the subset $\mset{S}_w$ and suppose we are given $n\Rrng$ independent and uniformly-distributed random bits. Proposition~4 in~\cite{boecherer-geiger-IT16} specifies that if $|\mset{S}_w| \le 2^{n\Rrng}$ then fixed-to-fixed length encoding gives
\begin{align}
  \diverg{P_{\rv{A}^n|\rv{W}}(\cdot|w)}{ Q^n_{\rv{A}|\binset{w}}}
  & \le \log_2\left( 1 + \frac{|\mset{S}_w|}{2\,q_{\rm min}(w)\, 2^{2n\Rrng}} \right) \nonumber \\
  & \le \frac{|\mset{S}_w|}{(2 \ln 2)\, q_{\rm min}(w)\, 2^{2n\Rrng}}
  \label{eq:synthesis-bound}
\end{align}
where $q_{\rm min}(w) = \min_{a^n\in\supp(P_{\rv{A}^n|\rv{W}}(\cdot|w))} Q^n_{\rv{A}|\binset{w}}(a^n)$.

It remains to bound $|\mset{S}_w|$ and $q_{\rm min}(w)$ and this is done in Theorem~\ref{thm:main-result} and Appendix~\ref{app:MLF-LLF-general-alphabets} below. The result is that $\Rrng$ can be made to vanish with growing $n$, and the \ac{I-divergence} \eqref{eq:I-div2} can be made to decay exponentially in $n$ for all $w=1,\dots,K$.

%-----------------------------
\subsection{Code for Minimum \ac{I-divergence}}
We next consider code design for the \ac{I-divergence} \eqref{eq:I-div1}.
\begin{proposition} \label{prop:best-set}
The code $\mset{S} = \supp(\ptarget)^n$ gives the smallest \ac{I-divergence} \eqref{eq:I-div1}.
\end{proposition}
\begin{proof}
Suppose $\mset{S} \subsetneq \supp(\ptarget)^n$ so that $q_{\mset{S}}=Q^n_\rv{A}(\mset{S})<1$. The encoder has sets $\mset{S}_w$ with probabilities $q_{\mset{S}_w}$. Now assign the unassigned strings with positive probability to obtain new sets $\mset{S}'_w$ with probabilities $q_{\mset{S}'_w}$ satisfying $q_{\mset{S}'_w} \ge q_{\mset{S}_w}$ and where at least one inequality is strict. We thus have
\begin{align}
    \diverg{U_\numW}{[q_{\mset{S}'_1},\ldots,q_{\mset{S}'_{\numW}}]} &= \sum_w \frac{1}{K} \log\frac{1/K}{q_{\mset{S}'_w}} \nonumber\\
    &< \diverg{U_\numW}{[q_{\mset{S}_1},\ldots,q_{\mset{S}_K}]}.
\end{align}
\end{proof}

Proposition~\ref{prop:best-set} shows that one should use all strings with positive probability if an ideal \ac{RNG} is available. Moreover, inflating $\mset{S}$ by strings outside $\supp(\ptarget)^n$ does not change the \ac{I-divergence}. 

%-------------------------
\subsection{Code Empirical Distribution and \ac{I-divergence}}
\label{subsec:Idiv-property}
The \ac{I-divergence} \eqref{eq:I-div} simplifies by applying \eqref{eq:crossentrp-identity} that one can interpret in terms of the code empirical \ac{pmf}. Let
\begin{align}
    \bar{n}_i = \sum_{a^n \in \mset{S}} P_{\rv{A}^n}(a^n) \, n_i(a^n)
    \label{eq:avg-occurrence}
\end{align}
be the average number of occurrences of letter $i$ in $\mset{S}$ and define the code empirical \ac{pmf} as
\begin{align}
  \bar{P}
  = \sum_{a^n \in \mset{S}} P_{\rv{A}^n}(a^n) \, \pi_{a^n}
  = \frac1n [\bar{n}_1,\dots,\bar{n}_{|\mset{A}|}] .
  \label{eq:Pavg}
\end{align}
Using \eqref{eq:X-H-D} and \eqref{eq:crossentrp-identity}, we have (see~\eqref{eq:div-convexity-bound})
\begin{align}
  \diverg{P_{\rv{A}^n}}{\qmf{A}^n}
  & = n \crossentrp{\bar{P}}{\qmf{A}} - \entrp{P_{\rv{A}^n}}.
  \label{eq:I-div-simplify}
\end{align}
We next use \eqref{eq:I-div-simplify} to analyze the performance of \acp{DM}.

%---------------------------
\section{Distribution Matching}
\label{sec:DM}

This section generalizes results of~\cite{schulte2017divergence} to non-binary alphabets. Recall that \ac{DM} is a special case of the model in Sec.~\ref{sec:model} where \eqref{eq:I-div2} is zero because there is no \ac{RNG}. Furthermore, from \eqref{eq:qusisplit} and \eqref{eq:I-div-simplify} we have
\begin{align}
  \diverg{P_{\rv{A}^n}}{\qmf{A}^n}
  & = \diverg{U_\numW}{\qmf{A}^n} \nonumber \\
  & = n \crossentrp{\bar{P}}{\qmf{A}} -\log_2 K \nonumber \\
  & = n\left(\crossentrp{\bar{P}}{\qmf{A}} - \Rinfo \right) \nonumber \\
  & = n\left(\entrp{\bar{P}} + \diverg{\bar{P}}{\qmf{A}} - \Rinfo \right).
  \label{eq:I-div-DM}
\end{align}

%---------------------------
\subsection{\ac{CCDM} Performance} \label{subsec:CCDM}
Consider the target \ac{pmf} $\qmf{A}=[q_1, \dots,q_{|\mset{A}|}]$ and a CCDM where all $a^n$ have the empirical \ac{pmf} $P=[p_1,\dots,p_{|\mset{A}|}]$ where $np_i$ is an integer for all $i$. We clearly have $\bar{P}=P$ and
\begin{align}
    & K = \binom{n}{np_1\ldots np_{|\mset{A}|}}, \quad
    \qmf{A}^n(a^n) = \prod_{i=1}^{|\mset{A}|} q_i^{np_i}
\end{align}
for all $a^n\in\mset{S}$ and the rate is
\begin{align}
  \Rinfo & = \frac{1}{n} \log_2 \binom{n}{np_1\ldots np_{|\mset{A}|}}.
  \label{eq:Rinfo-CCDM}
\end{align}
The bounds~\eqref{eq:multinomial-bound1} imply
\begin{align}
    \frac{|\mset{A}|-1}{2n} \log_2(2\pi n \, c)
    \le \entrp{P} - \Rinfo
    \le \frac{|\mset{A}|-1}{2n} \log_2 (8 n c)
   \label{eq:Aary-CCDM-D-bound}
\end{align}
where
\begin{align}
    c = \left( \prod_{i=1}^{|\mset{A}|} p_i\right)^{1/(|\mset{A}|-1)}
    \label{eq:Aary-CCDM-c}
\end{align}
and hence $\Rinfo\rightarrow\entrp{P}$ for large $n$. Moreover, we obtain $\entrp{P}\rightarrow\entrp{\qmf{A}}$ by choosing $P$ appropriately. For example, Algorithm 1 of~\cite{boecherer-geiger-IT16} gives a $P$ with $\done{P}{\qmf{A}}\le |\mset{A}|/(2n)$ (for $|\mset{A}|=2$, we obtain $p=\lfloor nq \rfloor/n$). Lemma~\ref{lemma:entropy-inequality} (or Lemma~\ref{lemma:entropy-inequality2}) gives the desired rate but by combining \eqref{eq:I-div-DM} and \eqref{eq:Aary-CCDM-D-bound} we have
\begin{align}
    & \frac{|\mset{A}|-1}{2} \log_2(2\pi n \, c) + \diverg{P}{\qmf{A}} \nonumber \\
    & \le 
    \diverg{\pblock}{\qmf{A}^n} \le \frac{|\mset{A}|-1}{2} \log_2(8 n c) + \diverg{P}{\qmf{A}} .
    \label{eq:Aary-CCDM-D-bound2}
\end{align}
The \ac{I-divergence} thus grows as $\frac12(|\mset{A}|-1)\log n$ with $n$ if $\diverg{P}{\qmf{A}}\rightarrow0$ or $\entrp{P}\rightarrow\entrp{\qmf{A}}$.

%---------------------------
\subsection{Improving \ac{CCDM}}
\label{subsec:CCDM-improve}
We consider two classes of \acp{pmf} for which the \ac{CCDM} pre-log factor $\frac12(|\mset{A}|-1)$ is suboptimal.

%---------------------------
\subsubsection{Product Distributions}
Suppose the target \ac{pmf} {\ptarget} splits into a product of \acp{pmf}:
\begin{align}
  \ptarget(a) = \ptarget(f(a',a'')) = Q_{\rv{A'}}(a') Q_{\rv{A''}}(a'')
  \label{eq:product-Q}
\end{align}
where $f$ is an invertible function.

\begin{example}
\label{example:product-Q}
Consider $\rv{A}=[\rv{A}_1,\rv{A}_2]$ where $\rv{A}_1$ and $\rv{A}_2$ are independent with \acp{pmf} $[p,1-p]$ and $[q,1-q]$, respectively. The 4-ary \ac{pmf} is $[pq,p(1-q),(1-p)q,(1-p)(1-q)]$.
\end{example}

For \acp{pmf} \eqref{eq:product-Q} one can use \ac{PDM} \cite[Sec. III]{steiner2018approaching,boecherer2017high,Pikus-Xu-CL17} that operates two or more component \acp{DM} in parallel. The \ac{I-divergence} \eqref{eq:I-div-DM} is then the sum of the \acp{I-divergence} of the component \acp{DM}, e.g., the \ac{PDM} pre-log factor for Example~\ref{example:product-Q} is $\frac12+\frac12 = 1$ while a 4-ary CCDM has the pre-log factor $\frac32$.

%---------------------------
\subsubsection{Unique Probabilities}
The best \ac{DM} for the target \ac{pmf} $\qmf{A}=U_{|\mset{A}|}$ has zero \ac{I-divergence} by putting out the $|\mset{A}|$-ary representation of $w-1$. We extend this observation to sources where $\qmf{A}$ and the empirical \acp{pmf} $P$ have the following form. Let $\mset{U}=\{1,\dots,|\mset{U}|\}$ enumerate the \emph{unique} probabilities in
\begin{align}
  P = [ \underbrace{p_1,\dots,p_1}_{\text{$\nu_1$ times}}, \;\dots\;,\underbrace{p_{|\mset{U}|},\dots,p_{|\mset{U}|}}_{\text{$\nu_{|\mset{U}|}$ times}} ]
\end{align}
where $p_j \ne p_k$ for $j\ne k$. The entropy is
\begin{align}
  \entrp{P} = \entrp{[r_1,\dots,r_{|\mset{U}|}]} + \sum_{j=1}^{|\mset{U}|} r_j \log_2 \nu_j
\end{align}
where $r_j = \nu_j p_j$ for $j=1,\dots,|\mset{U}|$. 

The key step now is as follows. Consider $a^n$ with empirical \ac{pmf} $P$, and consider the $n r_j$ positions where there are letters with empirical probability $p_j$. For these positions, we expand the \ac{CCDM} set $\mset{S}$ to include all $a^n$ with any of the $\nu_j^{r_j n}$ patterns of $\nu_j$ letters. These new strings all have the same probability $\qmf{A}^n(a^n)$. The new \ac{DM} again has $\bar{P}=P$ but now
\begin{align}
    & K = \binom{n}{nr_1\ldots nr_{|\mset{U}|}}
    \cdot \prod_{j=1}^{|\mset{U}|} \nu_j^{r_j n}, \quad
    \qmf{A}^n(a^n) = \prod_{j=1}^{|\mset{U}|} q_j^{nr_j}
\end{align}
for all $a^n\in\mset{S}$ and therefore
\begin{align}
  & \Rinfo
  = \frac{1}{n} \log_2 \binom{n}{nr_1\ldots nr_{|\mset{U}|}} + \sum_{j=1}^{|\mset{U}|} r_j \log_2 \nu_j.
  \label{eq:Rinfo-DM}
\end{align}
Equations~\eqref{eq:Aary-CCDM-D-bound} and~\eqref{eq:Aary-CCDM-c} are therefore updated as follows:
\begin{align}
    \frac{|\mset{U}|-1}{2n} \log_2(2\pi n \, c)
    \le \entrp{P} - \Rinfo
    \le \frac{|\mset{U}|-1}{2n} \log_2 (8 n c)
   \label{eq:Aary-DM-D-bound}
\end{align}
where
\begin{align}
    c = \left( \prod_{j=1}^{|\mset{U}|} r_j\right)^{1/(|\mset{U}|-1)}.
\end{align}
Hence we again have $\Rinfo\rightarrow\entrp{P}$ for large $n$ and we can make $\entrp{P}\rightarrow\entrp{\qmf{A}}$ by choosing $P$ appropriately as for \ac{CCDM}. Using the same approach as for \eqref{eq:Aary-CCDM-D-bound2} we further obtain
\begin{align}
  & \frac{|\mset{U}|-1}{2} \log_2(2\pi n \, c) + \diverg{P}{\qmf{A}} \nonumber \\
  & \le \diverg{\pblock}{\qmf{A}^n} \le \frac{|\mset{U}|-1}{2} \log_2(8 n c) + \diverg{P}{\qmf{A}}.
  \label{eq:Aary-DM-D-bound2}
\end{align}
The \ac{I-divergence} now grows as $\frac12(|\mset{U}|-1)\log n$ with $n$ rather than $\frac12(|\mset{A}|-1)\log n$ if $\diverg{P}{\qmf{A}}\rightarrow0$ or $\entrp{P}\rightarrow\entrp{\qmf{A}}$.

\begin{example}
Consider the \ac{pmf} $\ptarget = [0.6,0.2,0.2]$ and strings  of length 5. The code $\mset{S}$ has all strings with empirical \acp{pmf} $[3,2,0]/5$, $[3,1,1]/5$, and $[3,0,2]/5$. We compute
\begin{align}
    K = \binom{5}{3,2}\cdot 2^2 = 40. 
\end{align}
The code size of the corresponding \ac{CCDM} is instead
\begin{equation}
    \binom{5}{3,1,1} = 20.
\end{equation}
For large $n$, the bounds \eqref{eq:Aary-CCDM-D-bound2} show that the \ac{I-divergence} of the new \ac{DM} scales as $\log_2 n$ rather than $\frac{3}{2}\log_2 n$ as for \ac{CCDM}.
\end{example}

%---------------------------
\subsection{Optimal \ac{DM} Codes}
\label{subsec:DM}
The following result generalizes~\cite[Lemma~5]{schulte2017divergence} to non-binary alphabets.
\begin{proposition} \label{prop:DMcode}
The \ac{DM} code $\mset{S}$ that minimizes $\diverg{U_K}{\qmf{A}^n}$ has all $a^n$ with at least a specified probability with respect to $\qmf{A}^n$, i.e., $\mset{S}$ has all $a^n$ satisfying $\qmf{A}^n(a^n)\ge2^{-nI}$ for some $I$. Alternatively, $\mset{S}$ has all strings $a^n$ satisfying
\begin{align}
  \frac1n \selfinfo{\qmf{A}^n}{a^n}
  = \crossentrp{\pi_{a^n}}{\qmf{A}} \le I.
\end{align}
\end{proposition}
\begin{proof}
Consider some values $\hat{I}$ and $I$ with $\hat{I}<I$. Define $\mset{S}=\mset{S}'\cup\mset{S}''$ where $\mset{S}' = \{a^n:\crossentrp{\pi_{a^n}}{\qmf{A}}\le\hat{I}  \}$ and $\mset{S}''$ has $\ell$ strings $a^n$ with $\crossentrp{\pi_{a^n}}{\qmf{A}}=I$. We thus have $K=|\mset{S}|=|\mset{S}'|+\ell$ and
\begin{align}
\diverg{U_\numW}{\qmf{A}^n} =
-\log_2(|\mset{S}'|+\ell) + \frac{n|\mset{S}'|}{|\mset{S}'| + \ell} \bar{I}+ \frac{n\ell}{|\mset{S}'| +  \ell} I
\end{align}
where $\bar{I}=\frac{1}{|\mset{S}'|} \sum_{a^n\in\mset{S'}} \crossentrp{\pi_{a^n}}{\qmf{A}}$ and therefore $\bar{I}<I$. Now consider $\ell$ as a continuous variable and compute
\begin{align}
\frac{\partial}{\partial\ell}\diverg{U_\numW}{\qmf{A}^n} &= -\frac{1}{(\ln2)(|\mset{S}'|+\ell)} + \frac{n|\mset{S}'|}{(|\mset{S}'| + \ell)^2} \Delta I\\
\frac{\partial^2}{\partial\ell^2}\diverg{U_\numW}{\qmf{A}^n} &=\frac{1}{(\ln2)(|\mset{S}'|+\ell)^2} - \frac{2n|\mset{S}'|}{(|\mset{S}'|+\ell)^3} \Delta I
\end{align}
with $\Delta I = I -\bar{I}> 0$. The first derivative is zero only at
\begin{align}
\ell_0 =& |\mset{S}'| \left( (\ln2) n \Delta I-1 \right)
\end{align}
which means that there is only one extreme point. Note that $\ell_0$ can be negative but is larger than $-|\mset{S}'|$. The second derivative at $\ell=\ell_0$ evaluates to
\begin{align}
\frac{\partial^2}{\partial\ell^2}\diverg{U_\numW}{\qmf{A}^n}\Big|_{\ell=\ell_0}
=&-\frac{1}{(\ln2)^3|\mset{S}'|^2 n^2 \Delta I^2}
\end{align}
which is negative and therefore the \ac{I-divergence} (assuming $\ell$ is continuous) is maximum at $\ell=\ell_0$.

We now find the integer $\hat{\ell}\in\{0,1,\dots,\ell_{\text{max}}\}$ that minimizes $\diverg{U_\numW}{\qmf{A}^n}$ where $\ell_{\text{max}}$ is the number of length $n$ code strings that have cross entropy $I$. We distinguish three cases.
\begin{itemize}
	\item $\ell_0\in [0,\ell_{\text{max}}]$:
	$\diverg{U_\numW}{\qmf{A}^n}$ increases with $\ell$ for $0\le \ell\le\ell_0$ and decreases with $\ell$ for $\ell_0\le \ell \le \ell_{\text{max}}$.
	\item $\ell_0< 0$:
	$\diverg{U_\numW}{\qmf{A}^n}$ decreases with $\ell$ for $0\le\ell\le\ell_{\text{max}}$.
	\item $\ell_0> \ell_{\text{max}}$:
	$\diverg{U_\numW}{\qmf{A}^n}$ increases with $\ell$ for $0\le\ell\le\ell_{\text{max}}$.
\end{itemize}
%All cases are depicted in Fig.~\ref{fig:opt:casesL0} and in all cases we have $\hat{\ell}=0$ or $\hat{\ell}=\ell_{\text{max}}$.
In all cases we have $\hat{\ell}=0$ or $\hat{\ell}=\ell_{\text{max}}$.
Thus, the best code has all strings up to cross entropy $\hat{I}$ or $I$.
\end{proof}

Proposition~\ref{prop:DMcode} is certainly not obvious, e.g., it implies that optimal \ac{DM} codes have \emph{all} strings of any empirical \ac{pmf} that they contain.

\begin{figure}[t]
    \centering
    \pgfplotsset{compat=newest}
\pgfplotsset{every axis legend/.append style={%
cells={anchor=west}}
}
\usetikzlibrary{arrows}
\tikzset{>=stealth'}

\begin{tikzpicture}[]
\definecolor{mycolor1}{rgb}{0.00000,0.44700,0.74100}%
\definecolor{mycolor2}{rgb}{0.85000,0.32500,0.09800}%
\definecolor{mycolor3}{rgb}{0.92900,0.69400,0.12500}%
\definecolor{mycolor4}{rgb}{0.49400,0.18400,0.55600}%
\footnotesize
\begin{axis}[width=\columnwidth,
grid=both,
height=\psplotheight,
xmin=0,
xmax=17,
xlabel={codebook size $|\mset{S}|$},
ymin=0.003,
ymax=4.4,
ylabel={I-divergence $\mathbb{D}(P_{\rv{A}^n}\Vert Q^n_{\rv{A}})$},
legend style={legend cell align=left,align=left,draw=white!15!black,anchor=north west,at={(0.03,0.98)},fill opacity=0.7,text opacity = 1,draw opacity = 1}
]\addplot[mycolor1,mark=*,mark options={solid}] coordinates {
(1.0, 0.2960023257751077)
(2.0, 1.4199660824969005)
(3.0, 1.5429915006830086)
(4.0, 1.4819479608577972)
(5.0, 1.3724162416426142)
(6.0, 1.958967338497537)
(7.0, 2.343421704795887)
(8.0, 2.6059117175795894)
(9.0, 2.7899806755909093)
(10.0, 2.921172749708766)
(11.0, 3.015374363055754)
(12.0, 3.4369250546729617)
(13.0, 3.784363015000084)
(14.0, 4.074232249152477)
(15.0, 4.31857642192795)
(16.0, 3.729875474301383)
};
\addlegendentry{$q=0.05$}
\addplot[mycolor2,mark=*,mark options={solid}] coordinates {
(1.0, 0.9378610145480919)
(2.0, 1.1891111848126834)
(3.0, 1.0212320741797245)
(4.0, 0.8147362699449792)
(5.0, 0.6179331920840765)
(6.0, 0.8553988543561193)
(7.0, 0.9905064816666971)
(8.0, 1.0659864402095707)
(9.0, 1.1046031338113576)
(10.0, 1.1194333964015857)
(11.0, 1.1184298912259703)
(12.0, 1.315190719664809)
(13.0, 1.4724218726871556)
(14.0, 1.5992567005783145)
(15.0, 1.7023043879273816)
(16.0, 1.3172366104741622)
};
\addlegendentry{$q=0.15$}
\addplot[mycolor4,mark=*,mark options={solid}] coordinates {
(1.0, 1.5082785963192933)
(2.0, 1.3798908886382373)
(3.0, 1.0854658186900625)
(4.0, 0.8156970347977097)
(5.0, 0.5809301691422419)
(6.0, 0.6665406802360256)
(7.0, 0.6931803424192755)
(8.0, 0.6873093271166537)
(9.0, 0.6626530410608322)
(10.0, 0.6268649199249747)
(11.0, 0.5844463735189249)
(12.0, 0.6834216879414603)
(13.0, 0.7579112521807816)
(14.0, 0.8138247182579179)
(15.0, 0.8554072253681664)
(16.0, 0.5589216194355977)
};
\addlegendentry{$q=0.23$}
\end{axis}

\end{tikzpicture}
    \caption{\ac{I-divergence} \eqref{eq:I-div1} vs. $|\mset{S}|$ for $n = 4$ and various $q$.}
    \label{fig:DMcode}
\end{figure}
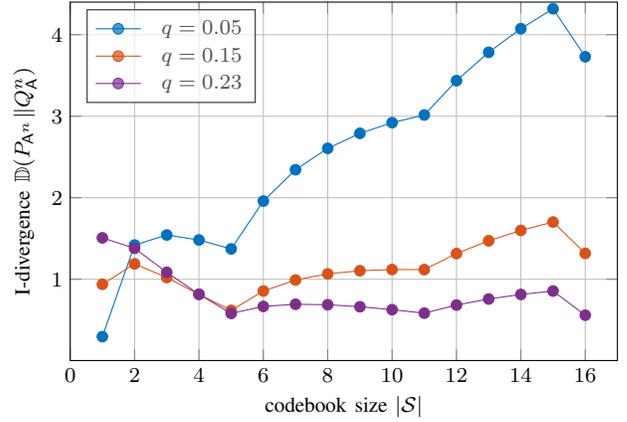

\begin{example}
For $|\mset{A}|=2$ there are only $n+1$ possible optimal codes although $K$ ranges from 1 to $2^n$. Fig.~\ref{fig:DMcode}
shows the \ac{I-divergence} behavior for a binary alphabet, block length $n=4$, and various $q$. The minimal \ac{I-divergence} is achieved at one of the $n+1=5$ values $K=1,5,11,15,16$.
\end{example}

Proposition~\ref{prop:DMcode} helps to prove the following basic result for binary strings.

\begin{theorem}[{\cite{schulte2017divergence}}] \label{theorem:DM-divergence}
Binary DM codes and encoders that minimize the \ac{I-divergence} have $\diverg{P_{\rv{A}^n}}{\qmf{A}^n}$ that grows as $\frac12 \log_2 n$ with $n$. Moreover, for binary alphabets \ac{CCDM} achieves this growth.
\end{theorem}
\begin{proof}
See Appendix~\ref{app:DM-binary}.
\end{proof}
%

%-------------------------
\section{MLF and LLF Algorithms}
\label{sec:MLF-and-LLF}
Since \ac{DM} cannot achieve low divergence in general, we now study one-to-many mappings. We propose two algorithms that generate encoder sets $\mset{S}_w$, $w=1,\dots,K$. Consider a code $\mset{S}$ and initialize $\mset{S}_w=\emptyset$ for all $w$. Order the strings in $\mset{S}$ from the most likely to the least likely. We consider two greedy approaches to populate the $\mset{S}_w$:
\begin{itemize}
\item \emph{Most-Likely First (\ac{MLF}):} Successively insert the \emph{most}-likely string into the set that has accumulated the least probability.
\item \emph{Least-Likely First (\ac{LLF}):} Successively insert the \emph{least}-likely string into the set that has accumulated the least probability.
\end{itemize}
The \ac{MLF} and \ac{LLF} approaches are specified in Algorithm~\ref{alg:MLF-LLF} where the choice of algorithm is reflected in steps 5 to 9.

%-------------------------
\begin{algorithm}[t]
\caption{MLF and LLF Algorithms}
\label{alg:MLF-LLF}
\begin{algorithmic}[1]
\Procedure{Partition}{$\mset{S}$, $K$, $\ptarget$, Algo}
\State $\mset{S}_w \gets \emptyset$,\quad $w=1,\ldots,K$
\State Sort $\mset{S}=\{a_1^n,\dots,a_{|\mset{S}|}^n\}$ so that $\qmf{A}^n(a_i^n) \ge \qmf{A}^n(a_j^n)$ for $i\le j$
%\Comment{}
\While{$\mset{S} \ne \emptyset$}
\If{Algo = MLF}
\State  $a^n\gets\text{first\_string}(\mset{S})$
\Else
\State $a^n\gets\text{last\_string}(\mset{S})$
\EndIf
\State $w\gets \arg \min_k \ptarget^n(\mset{S}_k)$
%\Comment{Find least probable set}
\State $\mset{S}_w \gets \mset{S}_w \cup \{a^n\}$
\State $\mset{S}\gets\mset{S}\setminus\{a^n\}$
\EndWhile
\State \textbf{return} $\mset{S}_1,\ldots,\mset{S}_K$ 
\EndProcedure
\end{algorithmic}
\end{algorithm}

%---------------------------------
\subsection{One Bit of Information per String}
Consider transmitting one bit of information so that \eqref{eq:I-div1} is
$\diverg{U_2}{[q_{\binset{1}},q_{\binset{2}}]}$. The \ac{MLF} and \ac{LLF} algorithms are not optimal in general. 

\begin{example}
Suppose the string probabilities are $[\frac{3}{10},\frac{2}{10}, \frac{1}{8},\frac{1}{8},\frac{1}{8},\frac{1}{8}]$. Both algorithms arrive at the set probabilities $[q_{\binset{1}},q_{\binset{2}}]=[11/20,9/20]$ but here it is best to group the strings to obtain $[q_{\binset{1}},q_{\binset{2}}]=[1/2,1/2]$.
\end{example}

The \ac{LLF} algorithm suggests a simple upper bound on $\diverg{U_2}{[q_{\binset{1}},q_{\binset{2}}]}$. The worst case has both sets equally likely just before inserting the last (most probable) string. For $\mset{S} = \supp(\ptarget)^n$, this string has $n$ occurrences of the most probable letter(s), i.e., its probability is $q_{\rm max}^n$ where $q_{\rm max}$ is the largest probability of any letter. The worst case \ac{pmf} is thus $[(1+q_{\rm max}^n)/2,(1-q_{\rm max}^n)/2]$ and we have
\begin{align}
  \diverg{U_2}{[q_{\binset{1}},q_{\binset{2}}]}
  & = \frac{1}{2} \log_2 \frac{1}{1-q_{\rm max}^{2n}} \nonumber \\
  & \le \frac{1}{2\ln2} \frac{q_{\rm max}^{2n}}{1-q_{\rm max}^{2n}}
  \label{eq:Dbound}
\end{align}
where the bound follows by $\ln(1+x)\le x$. The relation~\eqref{eq:Dbound} means that we can encode one bit of information with exponentially decreasing \ac{I-divergence} in $n$. However, from \eqref{eq:Rinfo+Rrng-lower} and $\Rinfo=1/n$ we find that the \ac{RNG} rate must satisfy $\Rrng\ge\Hrng\ge\entrp{\qmf{A}}-1/n$.

%-----------------------------
\subsection{MLF Encoder Properties}
We first develop a special property of the \ac{MLF} algorithm.

\begin{definition}[Pareto-optimal sets]
The assignment of strings to sets is Pareto-optimal if moving any individual string from one set to another does not decrease the \ac{I-divergence} \eqref{eq:I-div1}.
\end{definition}

\begin{proposition}
The \ac{MLF} algorithm generates Pareto-optimal sets. The \ac{LLF} algorithm does not generate Pareto-optimal sets in general.
\end{proposition}
\begin{proof}
Consider first the \ac{LLF} algorithm with the ordered string probabilities $[0.8,0.1,0.1]$ and $K=2$. \ac{LLF} assigns the third and first strings to one set and the second string to the other. By moving the third string (i.e., the first string that \ac{LLF} assigns) to the second set we obtain a better encoder.

Consider next the \ac{MLF} algorithm. Moving $a^n$ from $\mset{S}_v$ to $\mset{S}_w$ is Pareto efficient for \eqref{eq:I-div1} if and only if
\begin{align}
    \log_2\frac{1}{(q_{\mset{S}_v}-\ptarget^n(a^n))(q_{\mset{S}_w}+\ptarget^n(a^n))} \le \log_2\frac{1}{q_{\mset{S}_v}q_{\mset{S}_w}}
\end{align}
which is equivalent to $q_{\mset{S}_v} - q_{\mset{S}_w}\ge\ptarget^n(a^n)$.
Assuming $q_{\mset{S}_v} > q_{\mset{S}_w}$, the difference $q_{\mset{S}_v} - q_{\mset{S}_w}$   is at most the probability of the last string $\tilde{a}^n$ that was assigned to $\mset{S}_v$. Otherwise $\tilde{a}^n$ would have been assigned to $\mset{S}_w$. Therefore, moving $\tilde{a}^n$ from $\mset{S}_v$ to $\mset{S}_w$ does not improve \ac{I-divergence} \eqref{eq:I-div1}. All other strings in $\mset{S}_v$ have at least the same probability as $\tilde{a}^n$.
\end{proof}

Next, consider the code $\mset{S}$ and let $p_i$, $i=1,\ldots,|\mset{S}|$, be the probabilities of the ordered strings in $\mset{S}$, i.e., we have $p_i\ge p_j$ for $i\le j$. Let $\Delta_i$ be the difference in probability of the most likely set and least likely set after the $i$th (most likely) string from $\mset{S}$ has been assigned to a message. 

\begin{lemma} \label{lemma:MLF}
\ac{MLF} has $\Delta_i \le p_1$ for all $i \ge 0$.
\end{lemma}
\begin{proof}
We have $\Delta_0=0\le p_1$ and proceed by induction. Suppose that $\Delta_{i-1}\le p_1$ and $i\ge1$.

Consider first the case $p_i\ge \Delta_{i-1}$ so that the set to which string $i$ is assigned will have the most accumulated probability. We thus have $\Delta_i\le p_i$ with equality if the probability of the two least likely sets was the same before assigning string $i$. But then $\Delta_i\le p_1$ by the string ordering.

Consider next the case  $p_i< \Delta_{i-1}$ so that the most likely set did not change. Now we have $\Delta_i\le \Delta_{i-1}$ with equality if the two least likely sets was the same before assigning string $i$. But then we have $\Delta_i\le p_1$ by the inductive hypothesis.
\end{proof}

%-----------------------------
\subsection{LLF Encoder Properties}
We begin with an observation concerning the \ac{LLF} encoder.

\begin{proposition}
\ac{LLF} assigns the $(K-i)$-th string of the ordered list to the set $\mset{S}_{(i\!\!\mod K) + 1}$.
\label{prop:LLF-assignment}
\end{proposition}
\begin{proof}
At any step of the \ac{LLF} Algorithm, the difference of the most likely set probability and the least likely set probability is at most the probability of the next string to assign. After the assignment, the least probable set becomes (one of) the most probable set(s). In case of a tie, we order the new set last among the most probable sets.
\end{proof}

We remark that \ac{LLF} lets the decoder calculate the position in the ordered list and apply a modulo operation on the list. Algorithms that can accomplish this task include enumerative source encoding \cite{cover1973enumerative} and shell mapping \cite{laroia1994optimal,khandani1993shaping}.

Let $\Delta_i$ again be the difference in probability of the most likely set and least likely set after the $i$th (least likely) string from $\mset{S}$ has been assigned to a message. 

\begin{lemma} \label{lemma:LLF}
\ac{LLF} has $\Delta_i \le p_{K-i+1}$ for all $i \ge 0$. In particular, \ac{LLF} has $\Delta_i \le p_1$ for all $i\ge0$.
\end{lemma}
\begin{proof}
Define $p_{K+1}=0$. We have $\Delta_0=0\le p_{K+1}$ and proceed by induction. Suppose $\Delta_{i-1}\le p_{K-i+2}$ and $i\ge1$.

Consider first the case $p_{K-i+1}\ge \Delta_{i-1}$ so that the set to which string $K-i+1$ is assigned will have the most accumulated probability. We thus have $\Delta_i\le p_{K-i+1}$ with equality if the probability of the two least likely sets was the same before assigning string $i$.

Consider next the case  $p_{K-i+1}< \Delta_{i-1}$ so that the most likely set did not change. Now we have $\Delta_i\le \Delta_{i-1}$ with equality if the two least likely sets was the same before assigning string $i$. But then we have $\Delta_i\le p_{K-i+1}$ by the string ordering.
\end{proof}

%-----------------------------
\subsection{Achievable Rates}
\label{subsec:MLF-LLF-entropy}
We next analyze the information rate and \ac{I-divergence} \eqref{eq:I-div1} of the \ac{MLF} and \ac{LLF} algorithms. This section treats binary alphabets for simplicity and Appendix~\ref{app:MLF-LLF-general-alphabets} treats general alphabets. Suppose that $\ptarget(1) = q = 1-\ptarget(0) < 1/2$.

To prove our main result in Theorem~\ref{thm:main-result} below, we will use the code $\mset{S}=\mset{T}_\epsilon(\qmf{A})$. However, to facilitate the development and to gain insight, consider first the code $\mset{S}$ of binary strings with at most $n-k$ zeros. This means that the most likely string has probability $p_1=q^k (1-q)^{n-k}$. We remark that Proposition~\ref{prop:best-set} lets one reduce the \ac{I-divergence} \eqref{eq:I-div1} by later assigning the remaining strings in $\supp(\ptarget)^n$. The \ac{LLF} algorithm fits naturally into this framework. The \ac{MLF} algorithm does not fit, strictly speaking, because if we begin the \ac{MLF} assignment with the strings with $n-k$ zeros and later add the remaining strings in $\supp(\ptarget)^n$ then we do not have an \ac{MLF} algorithm. We will see, however, that the remaining strings can have an accumulated probability that vanishes exponentially in $n$, so the distinction makes little difference.

Consider the code $\mset{S}$ as specified and the \ac{iid} string $\rv{A}^n$ where $\rv{A}_1$ has \ac{pmf} $\qmf{A}$. We may write
\begin{align}
  q_{\mset{S}} = \Pr\left[ \rv{S}_n \le \frac{n-k}{n} \right]
\end{align}
where $\rv{S}_n=\sum_{i=1}^n (1-\rv{A}_i)/n$. Lemma~\ref{lemma:hoeffding} with $(k-1)/n<q$ gives the exponentially decaying bound
\begin{align}
    1-q_{\mset{S}}
    %&= \sum_{i = 0}^{k-1} \binom{n}{i} q^i (1-q)^{n-i} \nonumber \\
    &= \Pr\left[ \rv{S}_n - (1-q) \ge q - \frac{k-1}{n} \right] \nonumber \\
    &\le \exp\left(-2n\left(q-\frac{k-1}{n}\right)^2\right).
    \label{eq:hoeffding-MLF-LLF}
\end{align}
By Lemmas~\ref{lemma:MLF} and~\ref{lemma:LLF}, we have $\Delta_{|\mset{S}|}\le q^{k}(1-q)^{n-k}$
and
\begin{align}
  & K \left(\max_w q_{\binset{w}} - q^{k}(1-q)^{n-k} \right) \nonumber \\
  & \quad \le q_{\mset{S}} \le K \left( \min_w q_{\binset{w}} + q^{k}(1-q)^{n-k} \right) . \label{eq:qsetbound}
\end{align}
We thus have 
\begin{align}
  & \frac{q_{\mset{S}}}{K}-q^{k}(1-q)^{n-k} \le q_{\binset{w}} \le \frac{q_{\mset{S}}}{K}+q^{k}(1-q)^{n-k}
  \label{eq:qsetbound2}
\end{align}
for all $w=1,\dots,K$ and
\begin{align}
  & \diverg{U_{\numW}}{[q_{\binset{1}},\ldots,q_{\binset{\numW}}}] \nonumber \\
  &\qquad\le \sum_{w} \frac{1}{\numW} \log_2\frac{1/\numW}{q_{\mset{S}}/\numW-q^{k}(1-q)^{n-k}} \nonumber \\
  &\qquad= \log_2\frac{1}{1-\left[(1-q_{\mset{S}})+\numW q^{k}(1-q)^{n-k}\right]} \label{eq:divergence_simple1} \\
  &\qquad\overset{(a)}{\le} 2\left[(1-q_{\mset{S}})+\numW q^{k}(1-q)^{n-k}\right] \label{eq:divergence_simple}
\end{align}
where $(a)$ follows if the term in square brackets is at most $1/2$, since $-\log_2(1-x)\le 2x$ if $0\le x\le1/2$.

Consider the two summands in \eqref{eq:divergence_simple}. We have already seen that the term $1-q_{\mset{S}}$ vanishes exponentially in $n$ as long as $q>(k-1)/n$. In particular, neglecting quantization issues, we set $k=n(q-\epsilon)$ for small $\epsilon$ and~\eqref{eq:hoeffding-MLF-LLF} gives
\begin{equation}
    1 - q_{\mset{S}} \le \exp\left(-2n\epsilon^2\right).
\end{equation}
Next, consider a $\delta$ with $0<\delta<1$ and choose $\numW$ so that
\begin{align}
   \numW q^{k}(1-q)^{n-k} = (1-\delta)^n.
\end{align}
Taking logarithms and normalizing, one obtains
\begin{align}
   \Rinfo
   = \crossentrp{\qmf{A}+[-\epsilon,\epsilon]}{\qmf{A}} + \log_2(1-\delta). 
   \label{eq:rate}
\end{align}
We combine \eqref{eq:divergence_simple}-\eqref{eq:rate}, choose small positive $\epsilon$ and $\delta$, and choose $n$ sufficiently large so that the term in square brackets in \eqref{eq:divergence_simple} is at most $1/2$.

More generally, we have the following result for binary and non-binary alphabets, but with a different code $\mset{S}=\mset{T}_\epsilon(\qmf{A})$. The reason for the change is to show that $\Rrng$ can be made to vanish with $n$.

\begin{theorem}\label{thm:main-result}
The \ac{MLF} and \ac{LLF} algorithms generate encoders with $\Rinfo\rightarrow\entrp{\qmf{A}}$, exponentially decaying $\diverg{P_{\rv{A}^n}}{\qmf{A}^n}$, and $\Rrng\rightarrow0$ in $n$ by using $\mset{S}=\mset{T}_\epsilon(\qmf{A})$.
\end{theorem}
\begin{proof}
See Appendix~\ref{app:MLF-LLF-general-alphabets}.
\end{proof}

%-----------------------------
\subsection{Discussion}
\label{subsec:MLF-LLF-discussion}
The two key steps to show that the \ac{MLF} and \ac{LLF} algorithms have $\Rinfo\rightarrow\entrp{\qmf{A}}$ and $\diverg{P_{\rv{A}^n}}{\qmf{A}^n}\rightarrow0$ for large $n$ are establishing that $q_{\mset{S}}$ is close to one, see \eqref{eq:hoeffding-MLF-LLF}, and that the $q_{\binset{w}}$ are close to $1/K$, see \eqref{eq:qsetbound2}. There are several codes and encoders that meet these requirements. For example, for binary alphabets one may choose the code above that has all strings with at most $n-k$ zeros. Alternatively, one may use $\mset{S}=\mset{T}_\epsilon(\qmf{A})$ as for Theorem~\ref{thm:main-result}. In both cases, one satisfies \eqref{eq:qsetbound2} by choosing any partition of $\mset{S}$ into $K$ subsets with $q_{\mset{S}_w}\approx1/K$ for all $w$. The \ac{MLF} and \ac{LLF} algorithms are two methods to accomplish this task.

%-------------------------
\section{\ac{I-divergence} Lower Bounds}\label{sec:I-div-lower-bounds}
Consider the target \ac{pmf} $\qmf{A}$ with $q_{\rm max}$ as the largest letter probability. Fig.~\ref{fig:Bin_above_below} shows an example of $K$ bins where $q_{\rm max}^n>1/K$. The $\nup$ blue bins have exactly one string whose probability is at least $1/K$; the $\ndwn=K-\nup$ red bins have accumulated at most $1/K$ in probability and have one or more strings. Let the subsets $\mset{S}_\uparrow$ and $\mset{S}_\downarrow$ collect all strings of the blue and red bins, respectively.

\begin{figure}
    \centering
    \newcommand{\drahBucket}[4]
{
	\filldraw[fill=#4!40!white,thick] (#1,#2) rectangle (#1+0.5,#2+#3);
	\draw[thick] (#1,#2) --++(0,3);
	\draw[thick] (#1,#2) --++(0.5,0)--++(0,3);
}

\newcommand{\drahProbs}[4]
{
	\filldraw[fill=#4!40!white,thick] (#1,#2) rectangle (#1+1,#2+#3);
}

\begin{tikzpicture}[scale=0.8,font=\footnotesize,xscale=1.2,yscale=0.75]

\foreach [count=\num] \height in {2.5,2,2,2}{
\pgfmathsetmacro{\x}{(\num-1)*0.8}
\drahBucket{\x}{0}{\height}{blue};
}
\begin{scope}[xshift=3.5cm]
\foreach [count=\num] \height in {1.2,1,1,0.8}{
\pgfmathsetmacro{\x}{(\num-1)*0.8}
\drahBucket{\x}{0}{\height}{red};
}
\end{scope}
\draw[|-|] (-0.4,0) --++(0,1.5) node[left]{$\frac{1}{K}$};
\draw[dashed] (-0.4,0) ++(0,1.5) --++(6.9,0);
\node at (0.2,-0.4) {$w=1$};
\node at (6.3,-0.4) {$w=K$};
\draw[decorate,decoration={brace,amplitude=5},-] (6.8,-0.8) -- node[below,yshift=-3]{$K=2^{n\Rinfo}$ bins} (-0.25,-0.8);

\draw[decorate,decoration={brace,amplitude=5},-] (0,3.2) -- node[above,yshift=3]{$\nup$ bins with $a^n\in\mset{S}_\uparrow$} (3.0,3.2);
%\draw[decorate,decoration={brace,amplitude=5},-] (3.5,3.2) -- node[above,yshift=3]{$\ndwn$} (6.5,3.2);

\end{tikzpicture}
    \caption{Example bins after after applying the \ac{MLF} Algorithm. The bin heights represent the bin probabilities $q_{\binset{w}}=\ptarget^n(\binset{w})$.} 
    \label{fig:Bin_above_below}
\end{figure}
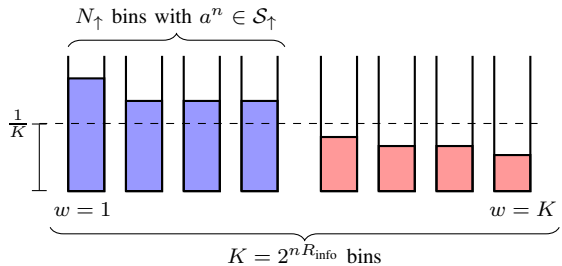
\begin{figure}
    \centering
    \newcommand{\drahBucket}[4]
{
	\filldraw[fill=#4!40!white,thick] (#1,#2) rectangle (#1+0.5,#2+#3);
	\draw[thick] (#1,#2) --++(0,3);
	\draw[thick] (#1,#2) --++(0.5,0)--++(0,3);
}

\newcommand{\drahProbs}[4]
{
	\filldraw[fill=#4!40!white,thick] (#1,#2) rectangle (#1+1,#2+#3);
}

\begin{tikzpicture}[scale=0.8,font=\footnotesize,xscale=1.2,yscale=0.75]

\begin{scope}[xshift= 12 cm]
\foreach [count=\num] \height in {2.125,2.125,2.125,2.125}{
\pgfmathsetmacro{\x}{(\num-1)*0.8}
\drahBucket{\x}{0}{\height}{blue};
}
\begin{scope}[xshift=3.5cm]
\foreach [count=\num] \height in {1.0,1,1,1.0}{
\pgfmathsetmacro{\x}{(\num-1)*0.8}
\drahBucket{\x}{0}{\height}{red};
}
\end{scope}
\draw[|-|] (-0.4,0) --++(0,2.125) node[left]{$\qup$};
\draw[|-|] (6.8,0) --++(0,1.0) node[right]{$\qdwn$};
\draw[dashed] (-0.4,0) ++(0,1.5) --++(6.9,0);
\draw[decorate,decoration={brace,amplitude=5},-] (3,-0.5) -- node[below,yshift=-3]{$\nup$ bins} (0.0,-0.5);
\draw[decorate,decoration={brace,amplitude=5},-] (6.6,-0.5) -- node[below,yshift=-3]{$\ndwn$ bins} (3.5,-0.5);

\end{scope}

\end{tikzpicture}
    \caption{Equalized bins.} 
    \label{fig:my_label}
\end{figure}
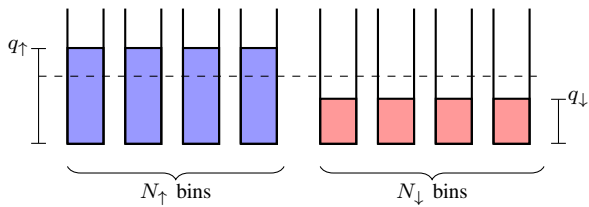

Define the \ac{pmf}
\begin{align}
  \bar{Q} = [ \underbrace{q_\uparrow,\dots,q_\uparrow}_{\text{$\nup$ times}}, \underbrace{q_\downarrow,\dots,q_\downarrow}_{\text{$\ndwn$ times}} ]
\end{align}
obtained by spreading the probability $q_{\mset{S}_\uparrow}$ equally over the $\nup$ bins with $a^n\in\mset{S}_{\uparrow}$, and similarly for the $\ndwn$ bins with $a^n\in\mset{S}_{\downarrow}$.
The convexity of \ac{I-divergence} (see Lemma~\ref{lemma:divergence-convexity}) implies
\begin{equation}
    \diverg{U_\numW}{[q_{\binset{1}},\ldots,q_{\binset{\numW}}]}
    \ge \diverg{U_\numW}{\bar{Q}}.
    \label{eq:simplification}
\end{equation}
We expand the right-hand side of \eqref{eq:simplification} as
\begin{align}
    \diverg{U_\numW}{\bar{Q}}
    & = \frac{\nup}{K} \log_2\frac{\frac{1}{K}}{\qup} + \frac{\ndwn}{K} \log_2\frac{\frac{1}{K}}{\qdwn} \nonumber \\
    & = \diverg{\left[\frac{N_\uparrow}{K},\frac{N_\downarrow}{K}\right]}{\left[q_{\mset{S}_\uparrow},q_{\mset{S}_\downarrow}\right]}.
    \label{eq:lowerboundPart2}
\end{align}
In fact, it is not necessary to group the first $\nup$ bins; any number $N'$ of grouped sets with $N'\le\nup$ and with accumulated probabilities $q_{\mset{S}_\uparrow}'$, $q_{\mset{S}_\uparrow}'$ works for the following result.

\begin{theorem}
The \ac{I-divergence} \eqref{eq:I-div1} generated with one-to-many mappings into disjoint sets satsifies
\begin{align}
    & \diverg{U_K}{[q_{\binset{1}},\ldots,q_{\binset{\numW}}]} \nonumber \\
    & \quad \ge \max_{N'\le \nup} \diverg{\left[\frac{N'}{K},1-\frac{N'}{K}\right]}{\left[q_{\mset{S}_\uparrow}',q_{\mset{S}_\downarrow}'\right]}.
    \label{eq:lb_final}
\end{align}
\end{theorem}

%------------------------
\subsection{Binary Alphabet}
Consider $|\mset{A}|=2$, $\mset{S} = \supp(\ptarget)^n$, and $q<1/2$. Let $k$ be the maximum integer for which $q^k(1-q)^{n-k} \geq 1/K=2^{-n\Rinfo}$, or equivalently 
\begin{equation}
    k = \left\lfloor n\cdot \frac{\log_2(1-q) + \Rinfo}{\log_2 (1-q)-\log_2 q}\right\rfloor \label{eq:kmax}.
\end{equation}
We then have
\begin{align}
    \nup &= \sum_{i=0}^k \binom{n}{i}\\
    q_{\mset{S}_\uparrow} &= \sum_{i=0}^k \binom{n}{i} q^{i} (1-q)^{n-i}. \label{eq:bin_pup}
\end{align}

Suppose $\Rinfo=h(q)$ which implies $k=\lfloor nq \rfloor$ according to \eqref{eq:kmax}. We use Lemmas~\ref{lemma:binomial-bound1} and~\ref{lemma:binomial-bound2} to obtain
\begin{equation}
    \frac{\nup}{K}
    \le \frac{1-q+\frac{2}{n}}{1-2q+\frac{1}{n}} \cdot\frac{1}{\sqrt{2\pi nq(1-q)} }
\end{equation}
which decreases as $1/\sqrt{n}$ in $n$ so that
\begin{equation}
    \lim_{n\rightarrow \infty} \frac{\nup}{K} = 0.%\frac{1-p}{(1-2p)\sqrt{2\pi np(1-p)}}
\end{equation}
For the probability $q_{\mset{S}_\uparrow}$, observe that the median of a binomial distribution is either $\lfloor nq \rfloor$ or $\lceil nq \rceil$ so $q_{\mset{S}_\uparrow}$ converges to $1/2$.  The \ac{I-divergence} \eqref{eq:lowerboundPart2} thus evaluates to 1 for large $n$ because the first \ac{pmf} converges to $[0,1]$ and the second \ac{pmf} converges to $[1/2,1/2]$. This means that the code cannot have \ac{I-divergence} \eqref{eq:I-div1} below 1 bit for large $n$. For $\Rinfo \le -\log_2(1-q)$ the lower bound is zero because $\mset{S}_\uparrow=\emptyset$.

%---------------------------
\section{Numerical Results}
\label{sec:numerical-results}

We evaluate the performance of the \ac{MLF} and \ac{LLF} algorithms with $\mset{S} = \supp(\ptarget)^n$. Fig.~\ref{fig:ach_extnfo_vs_divergence} plots the \ac{I-divergence} \eqref{eq:I-div1} against $\Rinfo$, as well as an upper bound based on \eqref{eq:divergence_simple1} and \eqref{eq:rate}, and the lower bound \eqref{eq:lb_final}. For all simulations, the target \ac{pmf} is $\qmf{A} = [0.11,0.89]$. We evaluate the lower bound~\eqref{eq:lb_final} for $N' = \sum_{i=0}^{k'} \binom{n}{i}$, where $k'$ is integer and $k'\leq k$.

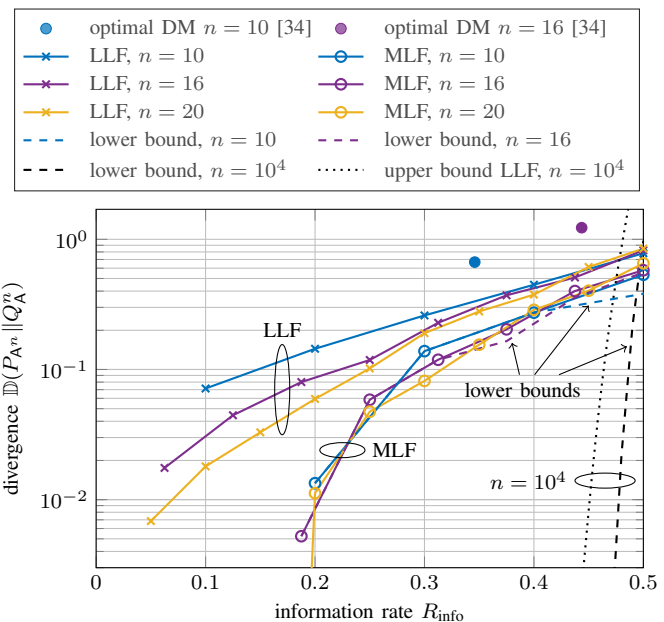
\begin{figure}[t]
    \centering
\definecolor{mycolor1}{rgb}{0.00000,0.44700,0.74100}%
\definecolor{mycolor2}{rgb}{0.85000,0.32500,0.09800}%
\definecolor{mycolor3}{rgb}{0.92900,0.69400,0.12500}%
\definecolor{mycolor4}{rgb}{0.49400,0.18400,0.55600}%
\footnotesize
\begin{tikzpicture}
\footnotesize
\begin{axis}[%
width=\columnwidth,
height=\psplotheight,
grid=both,
at={(0.758in,0.481in)},
%scale only axis,
xmin=0.0,
xmax=0.5,
xlabel={information rate $\Rinfo$},
ymin=0.003,
ymax=1.7,
ylabel={divergence $\mathbb{D}(P_{\rv{A}^n}\Vert Q^n_{\rv{A}})$},
ymode=log,
axis background/.style={fill=white},
legend columns=2, 
legend style={legend cell align=left,align=left,draw=white!15!black,anchor=south west,at={(-0.15,1.05)},fill opacity=0.7,text opacity = 1,draw opacity = 1}
]

%\addplot [color=mycolor2,dashed,thick]
%      table[col sep=tab] {figures/draw_divg_lb_n1000_0.89.txt};
%\addlegendentry{lower bound, $n=1000$};

%\addplot [color=mycolor2,dotted,thick]
%      table[col sep=tab] {figures/draw_divg_upperBound_n1000.txt};
%\addlegendentry{upper bound LLF, $n=1000$};

%\addplot [color=black,dotted, thick]
%  table[row sep=crcr]{%
%0.168	0\\
%0.168	0.7\\
%};
%\addlegendentry{$\Rinfo=\log_2(1/p)$};

\addplot [color=mycolor1,mark=*, only marks,mark options={fill=mycolor1}]
  table[row sep=crcr]{%
0.3459431618637	0.6681166142464\\
};
\addlegendentry{optimal DM $n=10$ \cite{boecherer2013block}};

\addplot [color=mycolor4,mark=*, only marks,mark options={fill=mycolor4}]
  table[row sep=crcr]{%
0.4436270051850329	1.2282332547336496\\
};
\addlegendentry{optimal DM $n=16$ \cite{boecherer2013block}};

\addplot [color=mycolor1,solid,thick,mark=x]
      table[col sep=tab] {figures/draw_divg_alg_LLF_n10_0.89.txt};
\addlegendentry{LLF, $n=10$};

\addplot [color=mycolor1,solid,thick,mark=o]
      table[col sep=tab] {figures/draw_divg_alg_MLF_n10_0.11.txt};
\addlegendentry{MLF, $n=10$};

\addplot [color=mycolor4,solid,thick,mark=x]
      table[col sep=tab] {figures/draw_divg_alg_LLF_n16_0.89.txt};
\addlegendentry{LLF, $n=16$};

\addplot [color=mycolor4,solid,thick,mark=o]
      table[col sep=tab] {figures/draw_divg_alg_MLF_n16_0.11.txt};
\addlegendentry{MLF, $n=16$};

\addplot [color=mycolor3,solid,thick,mark=x]
      table[col sep=tab] {figures/draw_divg_LFL_n20_0.89.txt};
\addlegendentry{LLF, $n=20$};

\addplot [color=mycolor3,solid,thick,mark=o]
      table[col sep=tab] {figures/draw_divg_MFL_n20_0.89.txt};
\addlegendentry{MLF, $n=20$};

%\addplot [color=mycolor1,dotted,thick]
%      table[col sep=tab] {figures/draw_divg_lb2_n10_0.89.txt};
%\addlegendentry{lower bound 2, $n=10$};

%-------------------------------------------------------------

\addplot [color=mycolor1,dashed,thick]
      table[col sep=tab] {figures/draw_divg_lb_n10_0.89.txt};
\addlegendentry{lower bound, $n=10$};

\addplot [color=mycolor4,dashed,thick]
      table[col sep=tab] {figures/draw_divg_lb_n16_0.89.txt};
\addlegendentry{lower bound, $n=16$};

%\addplot [color=mycolor4,dotted,thick]
%      table[col sep=tab] {figures/draw_divg_ub_n16_0.89.txt};
%\addlegendentry{upper bound, $n=16$};

%------------------------------------------------------------

%\addplot [color=mycolor2,solid,thick]
%      table[col sep=tab] {figures/draw_divg_alg_n20_0.89.txt};
%\addlegendentry{Alg. 2, $n=20$};

%------------------------------------------------------------

\addplot [color=black,dashed,thick]
      table[col sep=tab] {figures/new_lb_n10000.txt};
\addlegendentry{lower bound, $n=10^4$};

\addplot [color=black,dotted,thick]
      table[col sep=tab] {figures/draw_divg_upperBound_n10000.txt};
\addlegendentry{upper bound LLF, $n=10^4$};

% Labels
%\node (tenk) at (0.395,0.015) {$n=10^4$};
\draw [black] (0.465,0.014) ellipse [x radius = 0.4cm, y radius = 0.1cm] node [left,xshift=-0.4cm]{$n=10^4$};
\node (lbs) at (0.39,0.07) {lower bounds};
\draw[->] (lbs) -- (0.38,0.16);
\draw[->] (lbs) -- (0.45,0.31);
\draw[->] (lbs) -- (0.485,0.15);
%
%\node (mlf) at (0.275,0.024) {MLF};
\draw [black] (0.225,0.024) ellipse [x radius = 0.3cm, y radius = 0.1cm] node [right,xshift=0.3cm]{MLF};
%
%\node (llf) at (0.14,0.14) {LLF};
\draw [black] (0.17,0.07) ellipse [x radius = 0.1cm, y radius = 0.6cm] node [above,yshift=0.6cm]{LLF};

\end{axis}
\end{tikzpicture}%
    \caption{\ac{I-divergence} \eqref{eq:I-div1} vs. $\Rinfo$ for $\qmf{A}=[0.11,0.89]$ and different block lengths $n$. Note that $\entrp{\qmf{A}}\approx 0.5$.}
    \label{fig:ach_extnfo_vs_divergence}
\end{figure}

Note that the \ac{MLF} and \ac{LLF} algorithms sort binary strings of length $n$ so their complexity grows exponentially in $n$. The simulation results are restricted to string lengths with $n\le20$. As a reference, we plot the \ac{I-divergence} of the optimal \acp{DM} for $n=10$ and $n=16$. Observe that \ac{MLF} outperforms \ac{LLF} and has the same \ac{I-divergence} as the lower bound for small rates.

\section{Conclusions and Outlook}
\label{sec:conclusions-outlook}
We showed that \ac{ILD} coding is possible at rates approaching the entropy of a target \ac{pmf} with exponentially decaying \ac{I-divergence} and vanishing \ac{RNG} rate in the block length. The key step was to introduce invertible one-to-many mappings. For such mappings, an encoder was proposed that first chooses a subset of strings followed by an \ac{RNG} that chooses a string from the subset. The first step uses subsets that are generated by either an \ac{MLF} or \ac{LLF} algorithm. The second step uses a good \ac{RC}.

An interesting direction for future work is designing practical algorithms that approach the performance predicted by the theory.

\section*{Acknowledgements}
The authors wish to thank Juan Diego Lentner Iba\~nez for comments on the paper. This work was supported by DFG grant {KR 3517/9-1}.

\bibliographystyle{IEEEtran}
\bibliography{IEEEabrv,conf-jnls,references}

\clearpage
%%%%%%%%%%%%%%%%%%%%%%%%%%%%%%%%%%%%%%%%%%%%%%%%%%%%%%
% Appendices
%%%%%%%%%%%%%%%%%%%%%%%%%%%%%%%%%%%%%%%%%%%%%%%%%%%%%%
\setcounter{section}{0}
\renewcommand{\thesection}{\Alph{section}}
\renewcommand{\thesubsection}{\arabic{subsection}}
\renewcommand{\appendix}[1]{%
  \refstepcounter{section}%
  \par\begin{center}%
    \begin{sc}%
      Appendix \thesection\par%
      #1%
    \end{sc}%
  \end{center}\nobreak%
}

%-----------------------------
\appendix{Proof of Lemma~\ref{lemma:entropy-inequality2}}
\label{app:inequalities}
Consider the \acp{pmf} $P$ and $Q=P+\Delta$ where $\sum_i \Delta(i)=0$ and $\sum_i |\Delta(i)|=d_1$. Define
\begin{align}
  \Delta_+ = \sum_{i:\,\Delta(i)>0} \Delta(i), \quad
  \Delta_- = \sum_{i:\,\Delta(i)<0} \Delta(i)
\end{align}
and observe that $\Delta_{+}=-\Delta_{-}=d_1/2$. We expand
\begin{align}
  & \entrp{P} - \entrp{Q} \nonumber \\
  & = -\diverg{P}{Q} + \sum_i \Delta(i) \log_2(P(i) + \Delta(i))
  \label{eq:app-ei2-1}
\end{align}
and using $p_{\rm min} \le P(i) \le p_{\rm max}$ and $|\Delta(i)|\le d_1/2$ we have
\begin{align}
  & \sum_{i: \Delta(i)>0} \Delta(i) \underbrace{\log_2(P(i) + \Delta(i))}_{\le \log_2(p_{\rm max} + d_1/2)}
  \le \frac{d_1}{2} \log_2(p_{\rm max} + d_1/2) \label{eq:app-ei2-2} \\
  & \sum_{i: \Delta(i)<0} \Delta(i) \underbrace{\log_2(P(i) + \Delta(i))}_{_{\ge \log_2(p_{\rm min} - d_1/2)}}
  \le -\frac{d_1}{2} \log_2(p_{\rm min} - d_1/2). \label{eq:app-ei2-3}
\end{align}
Now insert \eqref{eq:app-ei2-2} and \eqref{eq:app-ei2-3} into \eqref{eq:app-ei2-1} and apply Lemma~\ref{lemma:divergence-bound}.

\medskip
%-----------------------------
\appendix{Proof of Theorem~\ref{theorem:DM-divergence}}
\label{app:DM-binary}

This appendix reviews results on binary \ac{DM} from~\cite{schulte2017divergence}. Consider $\mset{A} = \lbrace 0,1 \rbrace$ and observe that Proposition~\ref{prop:DMcode} lets one restrict attention to the $n+1$ code books $\mset{S}$ consisting of all strings with weight at most $k$ for $0 \le k \le n$. We have
\begin{align*}
	K = |\mset{S}| = \sum_{i = 0}^{k} \binom{n}{i}.
\end{align*}
The fraction of $1$s in $\mset{S}$ is
\begin{align}\label{eq:popt}
  \bar{p} =\frac{\sum_{i=0}^{k}\binom{n}{i}i}{|\mset{S}|n}
\end{align}
which increases monotonically in $k$ and reaches its maximum $\bar{p}=1/2$ for $k=n$. Let $p=k/n$. The next lemma shows that $\bar{p}\rightarrow p$ for large $n$ as long as $p<1/2$.

\begin{lemma}\label{lemma:PropertiesOfPck}
	For every positive integer $n$ and every integer $k$, $0\le k<n/2$, we have
	\begin{equation}\label{eq:pletter_bounds_in_k}
	0 \le  \frac{k}{n} - \bar{p}  \le \frac{1-k/n}{n(1-2k/n)} + \frac{1}{2n^2(1-2k/n)^2}.
	\end{equation}	 	
\end{lemma}

\begin{proof}
The lower bound is trivial. For the upper bound, we use Lemma \ref{lemma:binomial-identity} to write
\begin{align}
  \bar{p} & = \frac12 - \frac{\sum_{i=0}^{k}\binom{n}{i}(\frac{n}{2} -i)}{\sum_{j=0}^{k} \binom{n}{j} n} \nonumber\\
  & =\frac12 - \frac{\frac{n-k}{2} \binom{n}{k}}{\sum_{j=0}^{k} \binom{n}{j}n} \nonumber\\
  &= \frac12 - \left(\frac12 - \frac{k}{2n} \right)\frac{ \binom{n}{k}}{\sum_{j=0}^{k} \binom{n}{j}}.
\end{align}
Let $p=k/n$ and insert the lower bound in \eqref{eq:binomial-bound2} to obtain
\begin{align}
  \bar{p} & \ge \frac12 - \left(\frac12 - \frac{p}{2} \right)\frac{1-2p+1/n}{1-p+1/n} \left(1+\frac{1}{n(1-2p)^2}\right)\notag\\
  & \ge \frac12 - \frac{1-2p+1/n}{2} \left(1+\frac{1}{n(1-2p)^2}\right)\notag\\
  & = p - \frac{1-p}{n(1-2p)} - \frac{1}{2n^2(1-2p)^2}
\end{align}
which establishes the upper bound of \eqref{eq:pletter_bounds_in_k}.
\end{proof}

Let $\bar{P}=[\bar{p},1-\bar{p}]$ and $\ptarget=[q,1-q]$ where $0<q<1/2$. Recall that~\eqref{eq:I-div-DM} gives
\begin{align}
  \diverg{\unif{K}}{\qmf{A}^n}
  = n \crossentrp{\bar{P}}{\qmf{A}} -\log_2 |\mset{S}|.
  \label{eq:div-expand}
\end{align}
For small $n$, the best $k$ may have $k\ge n/2$. For example, Fig.~\ref{fig:DMcode} shows that $k=4$ gives the maximum $\Rinfo=1$ and the minimum $\diverg{U_K}{\qmf{A}^n}$ for $q=0.23$ and $n=4$. However, the following lemma shows that $k\ge n/2$ is not interesting for large $n$.

\begin{lemma}
\label{lemma:DM-binary-p1}
$\diverg{U_K}{\qmf{A}^n}$ grows linearly with $n$ if
\begin{align}
  \frac{k}{n} > p_1 := \frac{1 + \log_2(1-q)}{-\log_2 q + \log_2(1-q)}
  \label{eq:DM-binary-p1}
\end{align}
and $p_1$ satisfies $q<p_1<1/2$.
\end{lemma}
\begin{proof}
As already stated, $\bar{p}$ increases with $p=k/n$. Now choose $p$ so that $\bar{p}=p_1$ so that $\crossentrp{\bar{P}}{\qmf{A}}=1$. For this $p$ and large $n$, we have $\frac1n \log_2 |\mset{S}|<h(p_1)<1$ and the \ac{I-divergence} \eqref{eq:div-expand} grows linearly with $n$. Increasing $p$ further gives $\crossentrp{\bar{P}}{\qmf{A}}>1$ and \eqref{eq:div-expand} also grows linearly with $n$. Moreover, if $p_1<p<1/2$ then Lemma~\ref{lemma:PropertiesOfPck} shows that $\bar{p}\rightarrow p$ and \eqref{eq:div-expand} grows linearly in $n$. The bounds $q<p_1<1/2$ follow by using $1>h(q)$ and by showing that $p_1$ increases with $q$ to $p_1=1/2$ when $q=1/2$.
\end{proof}

Recall that \ac{CCDM} achieves $\frac12 \log_2 n$ growth, see Sec.~\ref{subsec:CCDM}. Lemma~\ref{lemma:DM-binary-p1} thus implies that we can focus on $k<np_1<n/2$ for large $n$. We remark that the bounds \eqref{eq:Rinfo+Rrng-upper} and \eqref{eq:Rinfo+Rrng-lower} imply that for $\frac1n\diverg{U_K}{\qmf{A}}\rightarrow 0$ we must have $\frac1n \log_2 |\mset{S}| \rightarrow h(q)$ and therefore $p\rightarrow q$ for large $n$.

Now for $p<1/2$, we obtain the following bounds from \eqref{eq:binomial-bound1} and \eqref{eq:binomial-bound2}:
\begin{equation}
|\mset{S}| \le \binom{n}{np} \frac{1-p}{1-2p}
\le \frac{2^{nh(p)}}{\sqrt{2\pi n p (1-p)}} \cdot \frac{1-p}{1-2p}.
\end{equation}
Inserting into~\eqref{eq:div-expand}, we have
\begin{align}
  \diverg{\unif{K}}{\qmf{A}^n}
  & \ge \frac12 \log_2 n - n \left[ h(p) - h(\bar{p}) \right]
  + n\diverg{\bar{P}}{\qmf{A}} \nonumber\\ 
  & \quad - \frac12 \log_2 \frac{1-p}{2\pi p (1-2p)^2} . \label{eq:div-expand2}
\end{align}
Define
\begin{align} 
  \epsilon(n) = \frac{1-p}{n(1-2p)} + \frac{1}{2n^2(1-2p)^2}.
\end{align}
For sufficiently large $n$, Lemmas~\ref{lemma:entropy-inequality2} and~\ref{lemma:PropertiesOfPck} give
\begin{align}\label{eq:entropy_Diff_optimalkn_vs_letter}
  & n \left[ h(p)- h(\bar{p}) \right] \nonumber \\
  & \le \left( \frac{1-p}{1-2p} + \frac{1}{2n(1-2p)^2} \right)
  \log_2\frac{1-p+\epsilon(n)}{p-\epsilon(n)}.
\end{align}
Since $\epsilon(n)\to 0$ for $n\to\infty$, we have
\begin{align}
  &\liminf_{n\to\infty} \left(\diverg{\unif{K}}{\qmf{A}^n} - \frac12 \log_2 n - n\diverg{[p,1-p]}{\qmf{A}}\right)\nonumber\\
  &\quad \ge - \frac{1-p}{1-2p}\log_2\frac{1-p}{p} - \frac12 \log_2 \frac{1-p}{2\pi p(1-2p)^2}.
\end{align}
The \ac{I-divergence} thus grows at least as $\frac12 \log_2 n$ with $n$. Moreover, $\ac{CCDM}$ achieves this growth by choosing $p=\lfloor nq \rfloor/n$ so that $d_1([p,1-p],\qmf{A})\le1/n$ and $n\diverg{[p,1-p]}{\qmf{A}}\rightarrow0$ for large $n$, see Sec.~\ref{subsec:CCDM}. Note that $\Rinfo\rightarrow h(q)$ for large $n$.

\medskip
%-----------------------------
\appendix{Proof of Theorem~\ref{thm:main-result}}
\label{app:MLF-LLF-general-alphabets}

This appendix extends the analysis of Sec.~\ref{subsec:MLF-LLF-entropy} to non-binary discrete alphabets. The key steps are to choose a code $\mset{S}$ with probability close to one and to show that all subset probabilities $q_{\binset{w}}$ are close to $1/K$.

Consider the code $\mset{S}=\mset{T}_\epsilon(\qmf{A})$. The left-hand side of \eqref{eq:typical-sets-1} in Lemma~\ref{lemma:typical-sets} gives
\begin{align}
    1-q_{\mset{S}}
    \le 2 |\mset{A}| \exp\left(-2n \, q_{\rm min}^2\, \epsilon^2\right).
    \label{eq:Hoeffding2}
\end{align}
By Lemmas~\ref{lemma:MLF} and~\ref{lemma:LLF}, we have $\Delta_{|\mset{S}|} \le \qmf{A}^n(a^n)$ for the $a^n$ with the largest probability in the typical set. For this $a^n$, we can bound (see~\eqref{eq:qsetbound2})
\begin{align}
  & \frac{q_{\mset{S}}}{K}-\qmf{A}^n(a^n) \le q_{\binset{w}} \le \frac{q_{\mset{S}}}{K}+\qmf{A}^n(a^n).
    \label{eq:MLF-LLF-subset-prob}
\end{align}
Following the same steps as in \eqref{eq:divergence_simple}, we have
\begin{align}
  \diverg{U_\numW}{[q_{\binset{1}},\ldots,q_{\binset{\numW}}]}
  \le 2\left[(1-q_{\mset{S}})+\numW \qmf{A}^n(a^n) \right] \label{eq:divergence_simple2}
\end{align}
if the scalar in square brackets is at most $1/2$.

Consider the two summands in \eqref{eq:divergence_simple2}. We have already seen that the term $1-q_{\mset{S}}$ vanishes exponentially in $n$.
Next, consider a $\delta$ with $0<\delta<1$ and choose $\numW$ so that
\begin{align}
   \numW \qmf{A}^n(a^n) = (1-\delta)^n.
   \label{eq:MLF-LLF-Kchoice}
\end{align}
Taking logarithms and normalizing, we have
\begin{align}
  \frac1n \log_2 \numW
  & = - \frac1n \log_2 \qmf{A}^n(a^n) + \log_2(1-\delta) \nonumber \\
  & \ge \entrp{\qmf{A}} (1-\epsilon) + \log_2(1-\delta) \label{eq:rate2}
\end{align}
where the inequality follows by the right-hand side of \eqref{eq:typical-sets-2} in~Lemma~\ref{lemma:typical-sets}. We thus choose
\begin{align}
  \Rinfo = \entrp{\qmf{A}} (1-\epsilon) - \gamma
  \label{eq:Rinfo-final}
\end{align}
where $\gamma=-\log_2(1-\delta)$, and \eqref{eq:divergence_simple2}-\eqref{eq:rate2} guarantee that this rate gives vanishing \ac{I-divergence} \eqref{eq:divergence_simple2}. Note that the term in square brackets in \eqref{eq:divergence_simple2} is less than 1/2 for large $n$. Finally, choose small positive $\epsilon$ and $\delta$ and large $n$ to complete the first part of the proof. 

Next, consider the \ac{RNG} and the bound \eqref{eq:synthesis-bound}. The bounds \eqref{eq:typical-sets-2} in Lemma~\ref{lemma:typical-sets} and $q_{\binset{w}} = \sum_{a^n\in\binset{w}} \qmf{A}^n(a^n)$ give
\begin{align}
  q_{\binset{w}} 2^{n\entrp{\qmf{A}}(1-\epsilon)}
  \le |\binset{w}|
  \le q_{\binset{w}} 2^{n\entrp{\qmf{A}}(1+\epsilon)}.
  \label{eq:Sw-bound}
\end{align}
The left-hand side of \eqref{eq:typical-sets-2} also gives
\begin{align}
  q_{\rm min}(w) \ge \frac{2^{-n\entrp{\qmf{A}}(1+\epsilon)}}{q_{\binset{w}}}.
  \label{eq:qmin-bound}
\end{align}
Inserting \eqref{eq:Sw-bound} and \eqref{eq:qmin-bound} into \eqref{eq:synthesis-bound}, we have
\begin{align}
  \diverg{P_{\rv{A}^n|\rv{W}}(\cdot|w)}{ Q^n_{\rv{A}|\binset{w}}}
  & \le \frac{q_{\binset{w}}^2 2^{2n\entrp{\qmf{A}}(1+\epsilon)}}{(2 \ln 2)\, 2^{2n\Rrng}} \nonumber \\
  & \overset{(a)}{\le} \frac{2^{4 n \epsilon \entrp{\qmf{A}}+2n\gamma+2}}{(2 \ln 2)\, 2^{2n\Rrng}}
  \label{eq:synthesis-bound2}
\end{align}
where step $(a)$ follows by applying \eqref{eq:MLF-LLF-subset-prob}, the right-hand side of \eqref{eq:typical-sets-2}, and \eqref{eq:rate2} to bound
\begin{align}
  q_{\binset{w}}
  & \le \frac{1}{K} + 2^{-n\entrp{\qmf{A}}(1-\epsilon)} \nonumber \\
  & \le 2 \cdot 2^{-n\entrp{\qmf{A}} + n \epsilon \entrp{\qmf{A}} + n \gamma} .
  \label{eq:synthesis-bound3}
\end{align}
We may thus choose
\begin{align}
    \Rrng = 2\epsilon\entrp{\qmf{A}} + 2\gamma
    \label{eq:Rrng-final}
\end{align}
which may vanish with $n$ because we can choose $\epsilon$ and $\gamma$ to vanish with $n$. Note that the rate \eqref{eq:Rrng-final} suffices for each $w$, i.e., one need not average over $w$ to achieve small $\Rrng$. 

Finally, both the \acp{I-divergence} on the left-hand sides of \eqref{eq:divergence_simple2} and \eqref{eq:synthesis-bound2} decay exponentially with $n$ if $\Rinfo$ and $\Rrng$ are given by \eqref{eq:Rinfo-final} and \eqref{eq:Rrng-final}, respectively. This implies hat $\diverg{P_{\rv{A}^n}}{\qmf{A}^n}$ decays exponentially with $n$.

\end{document}